\documentclass[11pt]{article}
\usepackage[margin=1in]{geometry}

\pdfmapline{=eufm7     eufm7     <eufm7.pfb}
\pdfmapline{=eufm10     eufm10     <eufm10.pfb}
\usepackage{lmodern}
\usepackage{anyfontsize}
\usepackage{hyperref}
\usepackage[T1]{fontenc}
\usepackage{bbm}
\usepackage{amsthm}
\usepackage{amsmath}
\usepackage{enumerate}
\usepackage{amssymb}
\usepackage[numbers,sort&compress]{natbib}
\usepackage{algorithm, algorithmic}
\usepackage{graphicx}
\usepackage{url}
\usepackage{multirow}
\usepackage{subfigure}
\usepackage{gensymb,textcomp}
\usepackage{tikz}
\usetikzlibrary{shapes,snakes}
\DeclareMathAlphabet{\mathscr}{OT1}{pzc}{m}{it} 

\newtheorem{theorem}{Theorem}
\newtheorem{corollary}[theorem]{Corollary}
\newtheorem{lemma}[theorem]{Lemma}
\newtheorem{claim}{Claim}

\newtheorem{definition}{Definition}

\newtheorem{observation}{Observation}
\renewcommand{\P}{{\cal P}} 
 
\newcommand{\etal}{{et al.~}}
\newenvironment{ntheorem}[1]{\medskip\noindent{\bf Theorem~#1.}\it}{\par}
\newenvironment{nlemma}[1]{\medskip\noindent{\bf Lemma~#1.}\it}{\par}

\newcommand{\argmax}{\operatornamewithlimits{argmax}}
\newcommand{\argmin}{\operatornamewithlimits{argmin}}

\newcommand{\tnx}{{\tilde{n}_x}}
\newcommand{\tns}{{\tilde{n}_s}}
\newcommand{\tni}{{\tilde{n}_i}}
\newcommand{\tno}{{\tilde{n}_o}}

\newcommand{\R}{\mathbb{R}}
\newcommand{\G}{\mathbb{G}}

\newcommand{\bA}{\mathbf{A}}

\renewcommand{\ae}{{a_\epsilon}}

\newcommand{\bP}{\mathbf{P}}
\newcommand{\bPi}{\mathbf{P}_i}
\newcommand{\bPo}{\mathbf{P}_o}

\newcommand{\btQ}{\tilde{\mathbf{Q}}}
\newcommand{\bQ}{\mathbf{Q}}
\newcommand{\bb}{\mathbf{b}}
\newcommand{\bc}{\mathbf{c}}
\newcommand{\bd}{\mathbf{d}}

\newcommand{\bqo}{\mathbf{q}_o}
\newcommand{\bqi}{\mathbf{q}_i}

\newcommand{\bx}{\mathbf{x}}
\newcommand{\bu}{\mathbf{u}}
\newcommand{\promise}{{\gamma}}
\newcommand{\bv}{\mathbf{v}}

\newcommand{\bvarsigma}{\boldsymbol{\varsigma}}
\newcommand{\vepsilon}{\xi}
\DeclareMathOperator{\poly}{poly}
\DeclareMathOperator{\polylog}{polylog}
\newcommand{\ppone}{\mbox{\sc Primal}}

\newcommand{\initial}{\mbox{\sc Initial}}

\newcommand{\bzero}{{\mathbf 0}}

\newcommand{\outero}{\mbox{\sc Outer}}
\newcommand{\spar}{\mbox{\sc Sparse}}
\newcommand{\lpdual}{\mbox{\sc Dual}}
\newcommand{\simp}{\mbox{\sc Switch}}
\newcommand{\innero}{\mbox{\sc Inner}}
\newcommand{\laginnero}{\mbox{\sc LagInner}}
\newcommand{\origa}{\boldsymbol{\mathscr{A}}}
\newcommand{\origb}{\boldsymbol{\mathscr{b}}}
\renewcommand{\Re}{\boldsymbol{\mathbb{R}}}
\newcommand{\bzeta}{\boldsymbol{\zeta}}
\newcommand{\by}{\mathbf{y}}
\newcommand{\bz}{\mathbf{z}}
\renewcommand{\L}{{\mathcal L}}
\newcommand{\M}{{\mathcal M}}

\newcommand{\C}{{\mathcal Cut}}
\newcommand{\K}{{\mathcal K}}
\renewcommand{\O}{\mathcal O}

\newcommand{\D}{\mathcal D}

\renewcommand{\t}[1]{\tilde{#1}}
\newcommand{\h}[1]{\hat{#1}}
\newcommand{\eps}{\varepsilon}

\newcommand{\set}{\mathscr{Set}}

\newcommand{\eat}[1]{}

\newcommand{\ki}{{k^*_i}}

\newcommand{\pos}{{\mathscr Pos}}

\newcommand{\hq}{\hat{q}}

\newcommand{\viol}{{\mathscr Viol}}

\newcounter{lp}
\newcommand{\lptag}{\tag{LP\arabic{lp}}\addtocounter{lp}{1}}

\newcommand{\bnorm}[1]{||#1||_b}

\newcommand{\sati}{{{\mathscr Sat}(i)}}

\newcommand{\satu}{{{\mathscr Sat}(U)}}
\renewcommand{\pos}{{\mathscr Pos}}

\title{Access to Data and Number of Iterations: Dual Primal Algorithms for Maximum Matching under Resource Constraints
\thanks{A preliminary extended abstract of this article appeared in SPAA 2015.}}
\date{}

\author{
Kook Jin Ahn\thanks{
This work was done while the author was at the Department of Computer and Information Sciences,
University of Pennsylvania, Philadelphia, PA 19104. Email: {\tt kookjin@cis.upenn.edu}.
The author is currently affiliated with Google Inc., 1600 Amphitheatre Parkway Mountain View, CA 94043. Email: {\tt kookjin@google.com}}.
\and
Sudipto Guha\thanks{Department of Computer and Information Sciences,
University of Pennsylvania, Philadelphia, PA 19104. Email: {\tt sudipto@cis.upenn.edu}.
Research supported in part by NSF Award CCF-1117216.}}

\begin{document}

\maketitle
\begin{abstract}
In this paper we consider graph algorithms in models of computation
where the space usage (random accessible storage, in addition to the
read only input) is sublinear in the number of edges $m$ and the
access to input data is constrained. These questions arises in 
many natural settings, and in particular in the analysis of MapReduce 
or similar algorithms that model constrained parallelism with 
sublinear central processing. In SPAA 2011, Lattanzi etal. provided a $O(1)$ approximation of maximum matching using $O(p)$ rounds of iterative filtering via mapreduce and $O(n^{1+1/p})$ space of central processing for a graph with $n$ nodes and $m$ edges. 
 
We focus on weighted nonbipartite maximum matching in this paper. For any constant $p>1$, we provide an iterative sampling based algorithm for computing a $(1-\epsilon)$-approximation of the
weighted nonbipartite maximum matching that uses $O(p/\epsilon)$
rounds of sampling, and $O(n^{1+1/p})$ space.  The results extends to $b$-Matching with small changes.
This paper combines 
adaptive sketching literature and fast primal-dual
algorithms based on relaxed Dantzig-Wolfe decision procedures.
Each round of sampling is implemented through linear sketches and executed in a single round of
MapReduce. The paper also proves 
that nonstandard linear relaxations of a problem, in particular penalty based formulations, are helpful in mapreduce and similar settings in reducing the 
adaptive dependence of the iterations.
\end{abstract}

\section{Introduction}
In many practical settings, such as map-reduce and its many variants, the overall framework
of an algorithm is constrained. To find a large maximum matching (the
actual edges and not just an estimate), the natural algorithm is
obvious: we iteratively sample edges, and show that the sampled edges
contain a large matching. Often these iterative algorithms converge
fast and provide a solution which is much better than the worst case
guarantees. One such example is the problem of weighted maximum
matching where Lattanzi \etal \cite{LattanziMSV11} showed that for a
$n$ node $m$ edge graph, we can find a $O(1)$ approximation using
$O(p)$ rounds of filtering and $O(n^{1+1/p})$ space (for any constant
$p>1$). However it was also shown that the approximation bound achieved 
was better than the worst case. This raises the natural question: 
is a $(1-\epsilon)$-approximation achievable without storing the entire 
graph in central processing?

\smallskip
Unfortunately, there are few systematic techniques that allow us to
analyze iterative algorithms. One well known example is linear (or
convex) programming -- and for maximum weighted matching there are LP
relaxations which are exact. Therefore any algorithm that provides a
$(1-\epsilon)$ approximation must also (possibly implicitly) provide a
bound for the underlying LP. 
 This raises the question: Can we analyze iterative algorithms for maximum matching? Note that
augmentation path based techniques either require random access or 
many (superconstant) iterations. However
the exact relaxation of matching has $m$ variables -- $y_{ij}$
indicating the presence of the edge $(i,j)$ in the matching. Moreover,
for the nonbipartite case the number of constraints is $2^n$,
corresponding to each odd set. The LP is given by \ref{lpbm} below. Since we can
address $b$--matching without much difficulty, we present that
version. Let $\bnorm{U}=\sum_{i \in U} b_i$ and $\O=\{U | \ \bnorm{U} \mbox{ is odd} \}$. 
For standard matching all $b_i=1$. Let $B=\sum_i b_i$. For a graph $G=(V,E)$ consider:

{\small
\begin{align*}
   &\displaystyle \beta^* =\max \sum_{(i,j) \in E} w_{ij} y_{ij}    \lptag\label{lpbm} \\
   &\displaystyle \sum_{j:(i,j) \in E} y_{ij} \leq b_i & \forall i \in V \\
   &\displaystyle \sum_{(i,j) \in E:i,j\in U} y_{ij} \leq \left\lfloor \bnorm{U}/2 \right\rfloor & \forall U \in \O \\
   & y_{ij} \geq 0 & \forall (i,j) \in E
\end{align*}
}
To achieve a $(1-\epsilon)$-approximation, the number of constraints in
\ref{lpbm} can be reduced to $n^{O(1/\epsilon)}$ by considering $\O_s
= \{ U \in \O| \ \bnorm{U} \leq 4/\epsilon \}$. However
$n^{O(1/\epsilon)}$ is still large.

There
has been an enormous amount of research on solving LPs efficiently
starting from Khachian's early result \cite{kha77}, for example, the
multiplicative weight update framework (\cite{LittlestoneW} and many
others), positive linear programming \cite{LubyN93}, fractional
packing and covering (\cite{PlotkinST95} and subsequent results),
matrix games \cite{GK95}, and many similar descriptions which exist in
different literature across different subfields (see the surveys
\cite{FosterV,AroraHK12}). None of the existing methods allow 
constant number of iterations. 

Each of these methods maintain multipliers (often referred to as weights in the literature, we use a different term since we consider matching in weighted graphs) on the constraints
and seek to optimize a linear combination (using the respective multipliers) of the constraints, thereby reducing multiple constraints to a single objective function. This is referred to as Dantzig-Wolfe type decomposition, since the resulting object is typically a simpler problem.  Methods which only maintain dual multipliers typically 
require $\Omega(\rho \epsilon^{-2}\log M)$ iterations for $M$ constraints, where $\rho$ is 
the width parameter (a variant of conditioning, defined shortly in the sequel). 
In fact this is a lower bound for random constraint matrices shown in \cite{KY99}; moreover the width parameter is a fundamental barrier. The width parameter of \ref{lpbm} is at least $n$. 
Methods such as \cite{prox1,prox2,BI04} which maintain both primal
and dual multipliers, are dominated by storing all the edges in the graph 
(same as number of primal variables)! Moreover most of these methods 
\cite{prox1,prox2} provide additive feasibility 
guarantees over an unit ball, and conversion to multiplicative error makes the number of iterations depend on the (square root of the) number of variables \cite{BI04}, a detailed discussion is available in Section~\ref{mirror}.

\medskip\noindent
{\bf Our Results:}  For any $\epsilon>0,p>1$ we provide a $(1-\epsilon)$ approximation scheme for the weighted nonbipartite matching problem using $O(p/\epsilon)$ rounds of adaptive sketching which can be implemented in MapReduce and $O(n^{1+1/p})$ centralized space. The space requirement increases to $O(n^{1+1/p} \log B)$ if $B=\sum_i b_i$ is super polynomial in $n$. The running time is $O(m\poly(\epsilon^{-1},\log n)\log B)$. From the perspective of techniques, we {\em partially} simulate multiple 
iterations of solving an LP in a single iteration -- this can be also viewed as defining an less adaptive method of solving LPs. However such a method naturally works for a subclass of LPs which we discuss next. In particular we focus on the dual of the LP we wish to solve, which leads to {\em Dual Primal Algorithms}.
In the context of matching the overall algorithm we get is very natural, given in Algorithm~\ref{alg:nat}.
\begin{algorithm}[H]
{\small
\begin{algorithmic}[1]
\vspace{0.05in}
\STATE Start with an initial sampling distribution over the edges.
\WHILE {we do not have a certificate of $(1-\epsilon)$ approximation}
\STATE Sample $O(n^{1+1/p})$ edges, subdivided into $t=O(\frac{1}{p\epsilon}\log n)$ independent parts, say $S_1,S_2,\ldots S_t$.
\STATE Use $S_1$ to simulate a step of solving a dual LP of matching. Then use $S_1$ to refine or adjust $S_2$ and then use $S_2$. In general, use $S_1,\ldots,S_q$ to refine $S_{q+1}$ and use $S_{q+1}$.
\STATE We prove that either we succeed in the refinement and use of $S_{q+1}$ or produce an explicit dual (of the dual) certificate -- which is a large explicit primal solution, in this case the desired matching.
\ENDWHILE

\end{algorithmic}
}
\caption{An algorithm for maximum matching\label{alg:nat}}
\end{algorithm}

\vspace{-0.1in}
The while loop in Algorithm~\ref{alg:nat} is executed for at most
$O(p/\epsilon)$ steps. Given the subdivision of each sample into
$O(\frac1{p\epsilon}\log n)$ parts -- we have an algorithm that uses
$O(\epsilon^{-2} \log n)$ iterations, but the adaptivity {\em at
the time of sampling} is only $O(p/\epsilon)$. However the adaptivity
{\em at the time of use} is still $O(\epsilon^{-2} \log n)$. This
differentiation is key and is likely to be use in many other
settings. One such setting (albeit in retrospect) is the linear sketch
based connectivity algorithm in \cite{AhnGM12,AhnGM12PODS}, where the
linear sketches\footnote{Linear Sketches are inner product of the input 
with suitable pseudorandom matrices, in this case the input is an oriented 
vertex-edge adjacency matrix. The sketch is computed first, and subsequently an adversary provides a cut. We then sample an edge across that cut (if one exists, or determine that no such edge exists) with high probability.} were computed in parallel in $1$ round {\em but used sequentially} in $O(\log n)$ steps of postprocessing to produce a spanning tree.
In this paper we show that $S_1,\ldots,S_t$ are relatives of
cut-sparsifiers. Cut sparsifiers, introduced by Benczur and Karger
\cite{BenczurK96}, are combinatorial objects that preserve every cut
to within $1\pm\epsilon$ factor. Use of cut sparsifiers is
nontrivial since:

{\small
$$ \sum_{(i,j) \in E:i,j\in U} y_{ij} \leq \left\lfloor \bnorm{U}/2 \right\rfloor$$}
is equivalent to sum and difference of cuts:
{\small
$$ \frac12 \sum_{i \in U} \left( \sum_{ (i,j) \in E} y_{ij} \right) - \left( \sum_{(i,j) \in E:i\in U, j \not \in U} y_{ij} \right) \leq \left\lfloor \bnorm{U}/2 \right\rfloor$$}

and no sparsifier can preserve differences of cuts approximately
(since that would answer the sign of the difference exactly). In fact,
it it easy to observe that the size of the largest matching in a graph
has no connection to large matchings in the sparsifier of that
graph.

\medskip
\noindent{\bf New Relaxations:}
 The dual of the standard relaxation for matching is \ref{bm-dual}, where the variables $x_i$ corresponds 
to the vertex constraints and $z_U$ correspond to the odd-sets. The dual multipliers (of this dual) corresponds to the edges.

{\small
\begin{align*}
& \beta^*=\min \sum_i b_i x_i + \sum_{U \in \O} \left\lfloor \bnorm{U}/2 \right\rfloor z_U \\
& \begin{array}{l l}
\displaystyle x_i+x_j+\sum_{U \in \O; i,j\in U} z_U\geq w_{ij} & \forall (i,j)\in E \\
   x_i, z_U \geq 0 & \forall i \in V, U \in \O\lptag \label{bm-dual} 
 \end{array}
\end{align*}
}

The width of the formulation \ref{bm-dual} is $2^n$ or $O(n^{1/\epsilon})$ -- there are 
no obvious ways of reducing the width. Consider now a different formulation of 
maximum matching ($w_{ij}=1$):

\vspace{-0.1in}
{\small
\begin{align*}
   &\displaystyle \max \sum_{(i,j) \in E} y_{ij} - 3 \sum_i \mu_i   \lptag\label{lpbm-un} \\
   &\displaystyle \sum_{j:(i,j) \in E} y_{ij} - 2 \mu_i \leq b_i &  \forall i \in V \\
   &\displaystyle \sum_{(i,j) \in E:i,j\in U} y_{ij} - \sum_{i \in U} \mu_i \leq \left\lfloor \bnorm{U}/2 \right\rfloor & \forall U \in \O \\
   & y_{ij},\mu_i \geq 0 &  \forall i \in V,\forall (i,j) \in E
\end{align*}
}
The above formulation allows each
vertex to be fractionally matched to $b_i + 2\mu_i$ edges, yet the
overall objective is charged for this flexibility -- this is a classic
penalty based formulation. It can be shown (and we do, for the general
weighted case, through ideas based on the proof of total dual
integrality) the objective function has not increased from \ref{lpbm}
(for $w_{ij}=1$). Use of \ref{lpbm-un} is more 
obvious if we consider its dual \ref{bm-dual-un}.

\vspace{-0.1in}
{\small
\begin{align*}
& \min \sum_i b_i x_i + \sum_{U \in \O} \left\lfloor \bnorm{U}/2 \right\rfloor z_U \\
& \begin{array}{l l}
\displaystyle x_i+x_j+\sum_{U \in \O; i,j\in U} z_U\geq 1 & \forall (i,j)\in E \\
\displaystyle 2x_i + \sum_{U:i \in U} z_U \leq 3 & \forall i \in V \\
   x_i, z_U \geq 0 & \forall i \in V, U \in \O \lptag \label{bm-dual-un} 
 \end{array}
\end{align*}
}
Note subject to $2x_i + \sum_{U:i \in U} z_U \leq 3$ and non-negativity
{\small $$ x_i+x_j+\sum_{U \in \O; i,j\in U} z_U\leq 6 $$}
or in other words, the width of the dual formulation is now independent of any problem parameters!
Therefore penalty based formulations are a natural candidate to study if we wish to add constraints to the dual; and such constraints may have to be added if we want to solve the dual faster.

However the biggest difficulty in implementing Algorithm~\ref{alg:nat}
arises from the step where we show that either we make large progress
in the dual or we can construct a large approximate matching. Note
that complementary slackness does not hold for approximate solutions,
so lack of improvement in the dual does not typically imply anything
for a primal solution.  However,we prove that {\em when we cannot make progress on the
dual} then the multipliers on the constraints (which are now assignment of values to primal variables, since we started with the dual) are such that a 
(weighted) cut-sparsifier, that treats the multiplier values on edges as weight/strength (this is not the edge weight in the basic matching problem) values, contains a large matching! This provides us an explicit sparse
subgraph containing a large matching as well as fractional matching
solution. 

We modify the linear sketch based algorithm in \cite{AhnGM12PODS} that
constructs cut-sparsifiers in a single round to over sample the edges
with probability by at most $(1+\epsilon)^t$ which is the maximum
amount by which the multiplier on an edge can change. For
$t=O(\frac1{p\epsilon}\log n)$ that bound is $n^{1/(2p)}$, and since
sparsification has $\tilde{O}(n)$ edges, the oversampled object would
have size $n^{1+1/p}$ (absorbing the terms polynomial in $1/\epsilon$
and $\log n$). 
The modified algorithm allows us {\em deferred evaluation/refinement}.

\medskip\noindent
{\bf Weighted nonbipartite graphs:} Finally, we show the new relaxation we need to consider for weighted non-bipartite graphs. Assume that the edge weights are at least $1$ and rounded to integral powers of $(1+\epsilon)$. Let $\h{w}_k=(1+\epsilon)^k$
and $\h{E}_k$ the set of edges $(i,j)$ with that weight. Then the (dual of the) maximum $b$--matching is given by \ref{fulllp}. Observe that \ref{fulllp} is very similar to \ref{bm-dual-un}, where we are considering a ``layered'' variant -- $x_{i(k)}$ corresponds to the cost of vertex $i$ in level $k$. $x_i = \max_k x_{i(k)}$ is the contribution of vertex $i$ to the objective. However the cost of each set $U$ in level $\ell$ is $z_{U,\ell}$ and the contribution of a set $U$ is {\em additive}! An edge $(i,j)$ is covered from the cost of the two vertices $i,j$ (specifically their cost in level $k$) or the sum of the costs of all $U$ (odd, containing both $i,j$) of all layers {\em below or equal $k$}.

{\small
\begin{align*}
& \beta^*=\min \sum_i b_i x_i + \sum_{U \in \O} \left\lfloor \frac{\bnorm{U}}{2} \right\rfloor \sum_{\ell} z_{U,\ell} \lptag  \label{fulllp} \\
& \displaystyle  x_{i(k)} + x_{j(k)} + \sum_{\ell \leq k} \left(\sum_{U \in \O; i,j\in U} z_{U,\ell} \right) 
\geq \h{w}_k 
& \forall (i,j)\in E_k\\
& \displaystyle  2 x_{i(k)} + \sum_{\ell \leq k} \left(\sum_{U \in \O:i \in U} z_{U,\ell}\right) \leq 3\h{w}_k & \forall  i,k \\
& x_i - x_{i(k)} \geq 0 & \forall i,k\\
&  x_{i(k)}, x_i, z_{U,\ell} \geq 0 & \forall i,k,U,\ell 
\end{align*}
}
The role of \ref{fulllp} is expressed by the following graph where all
$b_i=1$. It is clear that to get a $(1-\epsilon)$ approximation we
must consider the odd set corresponding to whole triangle, the bipartite relaxation has value $1+5\epsilon$ whereas the integral solution has value $1$. An assignment of $z_U=1$ to the entire graph using the relaxation \ref{bm-dual} is a valid solution, but the width parameter becomes $O(1/\epsilon)$ based on the 
edge of weight $10\epsilon$. Of course in this case 
we can set $z_U=10\epsilon$ and $x_i=1-10\epsilon$ for the vertex at the apex. 
However we not only need a more systematic relaxation 
for arbitrary graphs, we also need the relaxation to allow 
efficient computation in small space!

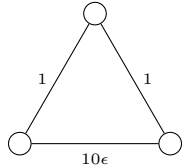
\begin{center}
\begin{tikzpicture}[font=\tiny]
 \tikzstyle{vertex}=[draw,shape=circle, minimum size=0.3cm,inner sep=2pt]
  \node[vertex] at (0,0) (u1) {};
 \node[vertex] at (2,0) (u2) {};
 \node[vertex] at (1,1.73) (u3) {};

\draw[-] (u1) to node[below] {$10\epsilon$}  (u2);
\draw[-] (u2) to node[right] {$1$}  (u3);
\draw[-] (u3) to node[left] {$1$}  (u1);
 \end{tikzpicture}
\end{center}

\vspace{-0.1in}
 It is not obvious why \ref{fulllp} should express maximum matching -- however as a  consequence, again, the width of the dual formulation is independent of any problem parameters. The dual of \ref{fulllp} is used in \ref{nicerlp} in Lemma~\ref{oolemma}, and the proof of \ref{fulllp} follows from the proof of Lemma~\ref{oolemma}. It is interesting that vertices and odd-sets are treated differently in \ref{fulllp} -- \ref{fulllp} is likely to be of interest independent of resource constraints.

\medskip\noindent
{\bf Related Work:} For weighted non-bipartite matching no previous result was known where the 
number of iterations is independent of the problem parameters. For unweighted cardinality matching
McGregor \cite{McGregor05} provided an algorithm with $2^{O(1/\epsilon)}$ iterations. Maximum matching is well studied in the context of {\bf bipartite} graphs,
see for instance
\cite{FeigenbaumKMSZ05,Zelke,EggertKS09,Epstein,LattanziMSV11,AhnG13BM,Kapralov13,BGM13}. The best known results in that context are either a $O(1)$ approximation using a single round
\cite{FeigenbaumKMSZ05,Zelke,EggertKS09,Epstein}, a $(1-\epsilon)$-approximation in $O(\epsilon^{-2}\log \frac1\epsilon)$ 
rounds \cite{AhnG13BM}, a $O(\epsilon^{-2})$ rounds in a vertex
arrival model (assuming order on the list of input edges)
\cite{Kapralov13}. The authors of \cite{GO13} show that a space bound of $n^{1 +
\Omega(1/p)}/p^{O(1)}$ is necessary for a $p$ round communication
protocol to find the exact maximum in bipartite graphs.

In the algorithm provided in this paper the probability 
of sampling each edge $(i,j)$ depends on (along with $i$ and $j$) the 
$z_U$ values of different odd sets $U$ of size at most $1/\epsilon$ (and containing both $i,j$). The number of such odd sets with $z_U > 0$ is at most 
$O(\epsilon^{-5}(\log B)(\log^2 n)\log^2 \frac1\epsilon)$. This is useful 
to show that the full $O(n^{1+1/p})$ space is not needed to define the value of the multiplier for an edge, specially in distributed settings.
The linear sketches can be viewed as requiring each vertex to 
sketch its neighborhood $n^{1/p}$ times. This has an obvious connections to
distributed computing, and in particular to the {\sc Congested Clique}
model \cite{Drucker2014}. Our linear sketch based result 
shows that in that model we
can compute a $(1-\epsilon)$ approximation for the maximum weighted
nonbipartite $b$--matching problem using $O(p/\epsilon)$ rounds and
$O(n^{1/p})$ size message per vertex. See \cite{MV13} for results in other distributed computation models.

\medskip\noindent{\bf Roadmap:} 
We provide the main definition and the 
main theorem (Theorem~\ref{useful}) about the dual-primal framework in Section~\ref{sec:general}. 
We then show how these
definitions and theorems are applied to matching in Section~\ref{sec:ex}. 
We discuss the construction of deferred cut-sparsifiers in Section~\ref{sec:deferred}, specially in resource constrained models. 
Section~\ref{init-weighted-bip} provides the initial solution, (note, of the dual).
We prove the new relaxation for matching in Section~\ref{proof:oolemma}.

\vspace{-0.1in}
\section{Dual-Primal Algorithms}
\vspace{-0.1in}
\subsection{The Framework}
\label{sec:general}

\begin{definition}\label{def:amenable}
The problem \ppone\ is defined to be
{\small
\[ \beta^* = \max \bc^T\by; \by \in \{ \origa^T\by \leq \origb, \by \geq \bzero  \} \tag{\ppone}\]
}
$(\rho_0,\rho_i)$--{\bf Dual-Primal amenable}  
if there exist $\bA \in
\Re_{+}^{m \times N}$, $\bb,\bx \in \Re_{+}^{N}$
$ \bP_o \in
\Re_{+}^{ \tno \times N}, \bqo \in \Re_{+}^{ \tno}$,
$\bP_i \in \Re_{+}^{\tni \times N}$, $\bqi \in \Re_{+}^{\tni}$,
convex polytope $\bQ$ with $\bzero \in \bQ$ and an absolute constant $a_1$ 
such that the following hold simultaneously:
\begin{enumerate}[(d1)]\parskip=0in
\item {\sc (Proof of Dual Feasibility.)} A feasible solution to $\{\bb^T\bx \leq \beta, \bA\bx \geq (1-3\epsilon)\bc, \bx \in \bQ,\bx\geq0\}$ implies $\beta^*\leq \beta/(1-a_1\epsilon)$.  
\label{dualfeasible}
\item {\sc (Outer Width.)} $\{\bP_o \bx \leq 2\bqo, \bx \in \bQ,\bx \geq 0 \}$ implies $\bA\bx \leq \rho_o \bc$.
\label{outerwidth}
\item  {\sc (Inner Width.)}  $\{\bP_i \bx \leq \bqi,\bx \in \bQ,\bx \geq 0 \}$ implies $\bP_o \bx \leq \rho_i \bqo$. \label{innerwidth}

\item {\sc (A deferred $\bu$-Sparsifier.)} Given $\bu,\bv \in \Re_{+}^m$, and the promise that all nonzero $\bu_\ell$ satisfy 
$1/\L_0 \leq \bv_{\ell}/\promise  \leq \bu_\ell \leq \promise \bv_{\ell}\leq \L_0$ 
for each $\ell \in [m]$ given fixed $\promise,\L_0 \geq 1$, 
we can construct a data structure $\D$ that samples a subset of indices in $[m]$ of size $\tns$ based on the $\bv$ values and stores the indices. 
After $\D$ has been constructed, the exact values of those 
stored entries of $\bu$ are revealed and the data structure constructs a 
nonnegative vector $\bu^s$ such that for some property $\G(
\bu^s,\bx)$ which is convex in $\bx$ and $\bx=\bzero$ satisfies $\G(
\bu^s,\bx)$, we have:
\vspace{-0.05in}
\begin{align*}
 & \hspace*{-1cm} (\bu^s)^T \bA \bx \geq \left(1-\frac\epsilon8\right) (\bu^s)^T\bc, ~~\G(\bu^s,\bx) \mbox{~and~} \bx \geq 0 
&  \implies \quad 
\bu^T \bA \bx \geq \left(1-\frac\epsilon2\right) \bu^T\bc  \quad \tag{\simp}
\end{align*}
\label{defcon}

\vspace{-0.2in}
\item {\sc (Initial solution.)} We can efficiently find a solution to \initial\ 
for some $a_\epsilon \geq 2$:
\[  
\left \{ 
\begin{array}{l} 
\bA\bx_0 \geq (1 - \epsilon_0) \bc \\
 \frac{\beta^*}{\ae} \leq \beta_0=\bb^T\bx_0 < \frac{\beta^*}{2} \\
 \bPo\bx_0 \leq 2\bqo \\
\bPi\bx_0 \leq \bqi \\
\bx_0 \in \bQ, 
\bx_0 \geq 0 \end{array} \right.
\tag{\initial} \]
\label{definit}
\end{enumerate}
\end{definition}
\noindent Observe that we can set $\rho_o=\rho_i=\infty$ but unbounded values of $\rho_o,\rho_i$ will not have any algorithmic consequence. For bounded $\rho_o,\rho_i$
the matrices $\bA$ and $\origa$ of course have to be related -- but that 
relationship is not necessarily a simple representation (such as column sampling). The main consequence
of $(\rho_0,\rho_i)$--dual primal amenability is as follows:

\begin{theorem}
\label{useful}
Suppose for a $(\rho_o,\rho_i)$--dual-primal amenable system \ppone\
there exists a {\sc MicroOracle} which
given: an absolute constant $a_2>0$, $\epsilon$ such that $0< \epsilon \leq 1$,
$\bu^s \in \Re^{m}_+,  \bzeta \in \Re^{\tno}_+$, and $ \varrho>0$,  
provides either:
\begin{enumerate}[(i)]\parskip=0in 
\item A feasible solution for \ppone\ with $\bc^T\by \geq (1-a_2\epsilon)\beta$ such that for all $\ell \in [m]$, $\by_\ell >0$ implies $\bu^s_\ell > 0$, i.e.,  a solution only involving the primal variables corresponding to the sampled constraints in the deferred sparsifier. 
\item Or a solution $\bx$ of \laginnero, where the number of non-zero entries of $\bx$ is at most $\tnx$.
\begin{align*}
& (\bu^s)^T \bA \bx - \varrho \bzeta^T \bP_o \bx \geq \left(1-\frac1{16}\right) (\bu^s)^T\bc - \varrho \bzeta^T \bqo \\
& \G(\bu^s,\bx) \tag{\laginnero}\\
& \btQ(\beta) = \{\bb^T\bx \leq \beta,\bPi \bx \leq \bqi,
\bx \in \bQ, \bx \geq 0\}
\end{align*}
\end{enumerate}
Then we can find a
$(1-(1 + a_1 + a_2)\epsilon)$-approximate solution to \ppone\ 
using $\tau$ rounds of deferred $\bu$-sparsifier construction, $\tau_o
=O\left(\rho_o\left(\frac1{\epsilon} + \epsilon \log \frac1{1-\epsilon_0}\right)\frac{\log (m/\epsilon)}{\log \promise} +  \frac{\log \ae}{\log \promise} \right)$.
 In each round we construct $O(\epsilon^{-1}\log \promise)$ deferred $\bu$--sparsifiers.
For a fixed deferred sparsifier {\sc MicroOracle} is invoked for at most $\tau_i = O(\rho_i (\log \rho_i) (\log \tno)\log \frac1\epsilon)$
times and we use a simple exponential of a linear combination of the returned solutions (of part (ii)) 
to define the weights $\bu$.

Each nonzero  $\bu_{\ell} \in  {\small \left[\left(2m/\epsilon\right)^{ - O\left(\frac{\rho_o}{\epsilon(1-\epsilon_0)} \right)},1 \right]}$ and $\log \L_0=O\left(\frac{\rho_o}{\epsilon(1-\epsilon_0)} \log (m/\epsilon) + \log \gamma \right)$. Moreover the algorithm computes 
(approximately) $\max \lambda; \bA\t{\bx} \geq \lambda \bc$ for a $\t{\bx}$ which is a weighted average of the $\bx$ returned by successive applications of part (ii) -- the algorithm can stop earlier than the stated number of $\tau$ rounds if it observes $\lambda \geq (1-3\epsilon)$. 
\end{theorem}

\vspace{-0.2in}
\begin{algorithm}[H]
{\small
\begin{algorithmic}[1]
\vspace{0.05in}
\STATE Start with $\bx=\bx_0$ and $\beta=\beta_0$, where $\bx_0,\beta_0$ refer to the initial solution.
\STATE Consider the following {\bf family} of decision problems:
{\footnotesize
\[
\left .
\begin{array}{ll} 
& \{\bA\bx \geq \bc, \bx \in \P(\beta) \}  \quad \mbox{where} \\  & \P(\beta)=  \left\{ \begin{array}{l}
\bP_o \bx \leq 2 \bqo  \\
\t{\bQ}(\beta) = \left\{
\bb^T\bx \leq \beta;
\bP_i \bx \leq \bqi; \bx \in \bQ;
\bx \geq \bzero \right \}
\end{array} \right.  
\end{array} \right. \tag{\lpdual}
\]
}
\vspace{-0.05in}
\WHILE {$\lambda < 1-3\epsilon$ (where $\lambda = \min_{\ell} ( \bA\bx)_{\ell}/\bc_{\ell}$) }
\STATE Define exponential weights $\bu$ (Section~\ref{proof:useful1}). Compute $\frac{\ln \gamma}{\epsilon}$ deferred $\bu$-sparsifiers denoted by $\{\D_q\}$. \label{samplestep}
\STATE Consider the union of the constraints sampled in the previous step and compute a $(1-a_3\epsilon)$ approximation to \ppone\ restricted to these constraints. Say that value is $\beta'$.
\label{approxstep}
\STATE If $\beta'>\beta(1-a_3\epsilon)/(1+\epsilon)$ remember the new solution and set 
$\beta = \beta'(1+\epsilon)/(1-a_3\epsilon)$.
\STATE Update $\bx,\bu(q)$ etc. as required by the proof of Theorem~\ref{useful}, we are guaranteed to not invoke condition (i) of {\sc MicroOracle} due to the Step~(\ref{approxstep}). This provides a new sampling weight function 
for Step~(\ref{samplestep}).
\ENDWHILE
\STATE Output the $\by$ corresponding to largest $\beta' \geq (1-a_3\epsilon)\beta/(1+\epsilon)$. We prove that such a $\beta'$ exists.
\end{algorithmic}
}
\caption{An algorithm for Theorem~\ref{useful} (Compare Algorithm~\ref{alg:nat}.)\label{alg:simplealg}}
\end{algorithm}

\noindent
The running time of the algorithm in Theorem~\ref{useful} is dominated by the time to construct the deferred sparsifiers (which includes the evaluation of $\bu$) plus the invocations of {\sc MicroOracle}. The space used to maintain the current average $\bx$ is $O(\tnx \tau_o \tau_i \frac1\epsilon \log \gamma)$. The space to store a single deferred sparsifier is typically  $O(n\promise^2 \poly(\epsilon^{-1}, \log n, \log \promise))$ (this is application specific and depends on $\bu,\L_0$). Therefore if $\rho_o,\rho_i$ are $poly(\epsilon^{-1})$, $\promise=n^{1/(2p)}$ and $\tnx=O(n \epsilon^{-1} \log n)
$, the overall asymptotic space complexity 
will typically be at most $n^{1+1/p}$ (absorbing the $\poly(\epsilon^{-1}, \log n)$ terms), which is $o(m)$ if $m \gg n^{1+1/p}$. For $b$--matching $\tnx=O(n \epsilon^{-1} \log B)$ where $B=\sum_i b_i$.

\medskip
\noindent {\bf Algorithmic Interpretations and Integral Solutions:} The algorithm behind Theorem~\ref{useful} can be expressed as a natural algorithm (relegating the details of how the $\bu$ are defined) as described in  Algorithm~\ref{alg:simplealg}. That algorithm worsens the guarantee to $(1-(1+a_1+a_2+a_3)\epsilon)$--approximation; but also provides an integral solution (assuming such an approximation algorithm against the LP relaxation exists as in Step~\ref{approxstep}) or we can set $a_3=0$ and settle for a fractional solution. For weighted $b$--Matching, such approximation algorithms which provide integral solutions exist \cite{DuanP10,AhnG14}. 
Finally observe the following corollary of Theorem~\ref{useful}.

\begin{corollary}
if we can construct (i) $\D$ used for the deferred $\bu$--sparsifier using $g$ rounds of sketching, and (ii) the initial solution using $k$ rounds of sketching and $n_{init}$ space,
then the algorithm in Theorem~\ref{useful} can be
implemented in $(g \tau + k)$ rounds of adaptive sketching. 
\end{corollary}

\noindent
{\bf The Intuition behind the Dual-Primal Setup:}
Suppose that the true dual weights at time $t$ is $\bu(t)$ then
primal-dual algorithm can be thought of as progressing through the
sequence
$\bu(1)$, $\t{\bx}(1)$, $\bu(2)$, $\t{\bx}(2)$, $\bu(3)$,$\t{\bx}(3)$,$\ldots$. If we
do not prove any further properties, we do not change the fundamental
dependence between $\bu(1),\bu(2),\ldots$ and the overall process with
sketching remains as adaptive as before, i.e., $\bu(1) \rightarrow
\bu(2) \rightarrow \bu(3)$ changes to $\bu(1) \rightarrow \bu^s(1)
\rightarrow \bu(2) \rightarrow \bu^s(2) \rightarrow \bu(3)$, {\em which
are identical from the perspective of the number of adaptive steps
required to compute} $\bu(t)$ or $\bu^s(t)$. This is shown in the left
part of the Figure~\ref{figa}. 

\begin{figure}[H]
\begin{center}
\begin{tikzpicture}[font=\tiny]
 \tikzstyle{vertex}=[draw,shape=circle, minimum size=0.3cm,inner sep=2pt]
 \tikzstyle{sket}=[draw, shape=regular polygon,regular polygon sides=9,inner sep=-1pt]
 \tikzstyle{rec}=[draw, rectangle] 
 \tikzstyle{prom}=[draw, shape=regular polygon,regular polygon sides=5,inner sep=-1pt]
  \node[vertex] at (0,3) (u1) {$\bu(1)$};
 \node[vertex] at (2,3) (u2) {$\bu(2)$};
 \node[vertex] at (4,3) (u3) {$\bu(3)$};
 \node[sket] at (0,1.5) (v1) {$\bu^s(1)$};
 \node[sket] at (2,1.5) (v2) {$\bu^s(2)$};
 \node[sket] at (4,1.5) (v3) {$\bu^s(3)$};
\draw[->] (u1) to node[left] {Sketch}  (v1);
\draw[->] (u2) to node[left] {Sketch}  (v2);
\draw[->] (u3) to node[left] {Sketch}  (v3);
\node[vertex] at (0,0) (w1) {$\t{\bx}(1)$};
 \node[vertex] at (2,0) (w2) {$\t{\bx}(2)$};
 \node[vertex] at (4,0) (w3) {$\t{\bx}(3)$};
\draw[->] (v1) to node[left] {Oracle}  (w1);
\draw[->] (v2) to node[left] {Oracle}  (w2);
\draw[->] (v3) to node[left] {Oracle}  (w3);
\draw[->] (w1) to node[below,xshift=0.3cm] {update} (1,0) ;
\draw[-] (1,0) to (1,3);
\draw[->] (1,3) to (u2);
\draw[->] (w2) to node[below,xshift=0.3cm] {update} (3,0);
\draw[-] (3,0) to (3,3);
\draw[->] (3,3) to (u3);
\draw[->] (5.25,1.5) to (6.75,1.5);

\node[prom, minimum size=0.7cm] at (7.3,2.5) {};
\node[prom, draw opacity=1, fill=white, fill opacity=1, minimum size=0.7cm ] at (7.4,2.5) {};
\node[prom, draw opacity=1, fill=white, fill opacity=1] at (7.5,2.5) (nu0) {$\bvarsigma(1)$};
 \node[vertex] at (9,3) (nu1) {$\bu(1)$};
 \node[vertex] at (11,3) (nu2) {$\bu(2)$};
 \node[vertex] at (13,3) (nu3) {$\bu(3)$};
 \node[rec] at (9,2.1) (ru1) {$\D(1)$};
 \node[rec] at (11,2.1) (ru2) {$\D(2)$};
 \node[rec] at (13,2.1) (ru3) {$\D(3)$};
 \node[sket] at (9,1.1) (nv1) {$\bu^s(1)$};
 \node[sket] at (11,1.1) (nv2) {$\bu^s(2)$};
 \node[sket] at (13,1.1) (nv3) {$\bu^s(3)$};
\draw[->,dashed] (nu1) to (ru1);
\draw[->,dashed] (nu2) to  (ru2);
\draw[->,dashed] (nu3) to (ru3);
\draw[->] (nu0) to (ru1);
\draw[-] (nu0) to (9.75,2.4);
\draw[->] (9.75,2.4) to (ru2);
\draw[-] (nu0) to (11.5,2.5);
\draw[->] (11.5,2.5) to (ru3);
\draw[->] (ru1) to node[left] {Refine}  (nv1);
\draw[->] (ru2) to node[left] {Refine}  (nv2);
\draw[->] (ru3) to node[left] {Refine}  (nv3);
\node[vertex] at (9,0) (nw1) {$\t{\bx}(1)$};
 \node[vertex] at (11,0) (nw2) {$\t{\bx}(2)$};
 \node[vertex] at (13,0) (nw3) {$\t{\bx}(3)$};
\draw[->] (nv1) to node[left] {Oracle}  (nw1);
\draw[->] (nv2) to node[left] {Oracle}  (nw2);
\draw[->] (nv3) to node[left] {Oracle}  (nw3);
\draw[->] (nw1) to node[below,xshift=0.3cm] {update} (10,0) ;
\draw[-] (10,0) to (10,2.1);
\draw[->] (10,2.1) to (ru2);
\draw[-,dashed] (10,2) to (10,3);
\draw[->,dashed] (10,3) to (nu2);
\draw[->] (nw2) to node[below,xshift=0.3cm] {update} (12,0);
\draw[-] (12,0) to (12,2.1);
\draw[->] (12,2.1) to (ru3);
\draw[-,dashed] (12,2) to (12,3);
\draw[->,dashed] (12,3) to (nu3);
 \end{tikzpicture}
\end{center}
\caption{Adaptivity in a dual primal algorithm.\label{figa}}
\end{figure}
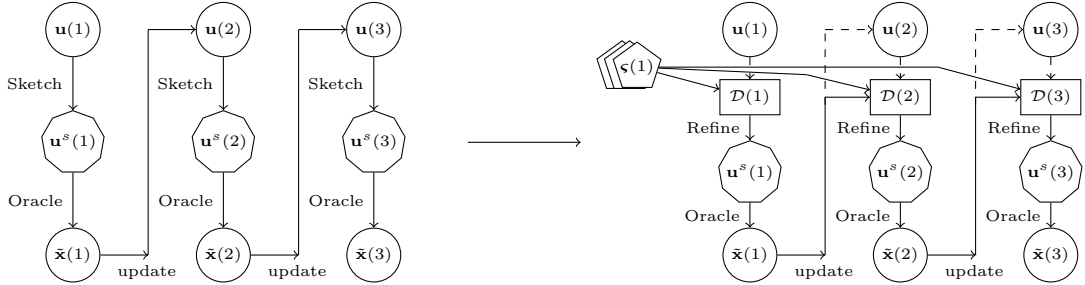

Deferred sparsifiers constructed via sampling allows us to bypass that dependency chain mentioned above. The sparsifier construction is broken into two parts -- the first part of the construction can be made non-adaptive and in parallel. Then the second part can be performed sequentially, but over a small subset of edges stored in memory.
We construct $\bvarsigma(1),\ldots,\bvarsigma(t)$ in parallel which eventually give us the deferred sparsifiers $\D(1),\ldots,\D(t)$. We derive the actual sparsifier $\bu^s(1)$ and the update $\t{\bx}(1)$. Instead of explicitly computing $\bu(2)$, we refine the weights of $\D(2)$ and produce $\bu^s(2)$ directly. Therefore we sequentially compute $\bu^s(1) \rightarrow \bu^s(2) \rightarrow \cdots \rightarrow \bu^s(t)$ and are able to make $t$ simultaneous steps without further access to data. This is shown in the right part of Figure~\ref{figa}. 
However such a reduction of adaptivity is only feasible for specific linear programming relaxations and we should choose that relaxation with care.

\subsection{Proof of Main Theorem}
\label{proof:useful}

In this subsection we prove Theorem~\ref{useful}, based on Definition~\ref{def:amenable}.

\begin{ntheorem}{\ref{useful}}
Suppose for a $(\rho_o,\rho_i)$--dual-primal amenable system \ppone\
given any constant $\epsilon$, $0< \epsilon \leq 1, \bu^s \in \Re^{m}_+,  \bzeta \in \Re^{\tno}_+,  \varrho>0$, a {\sc MicroOracle}  
provides either:
\begin{enumerate}[(i)]\parskip=0in 
\item A feasible solution for \ppone\ with $\bc^T\by \geq (1-a_2\epsilon)\beta$ such that for all $\ell \in [m]$, $\by_\ell >0$ implies $\bu^s_\ell > 0$  
-- i.e.,  a solution only involving the primal variables corresponding to the sampled constraints in the deferred sparsifier. 
\item Or a solution $\bx$ of \laginnero, where the number of non-zero entries of $\bx$ is at most $\tnx$.
\begin{align*}
& (\bu^s)^T \bA \bx - \varrho \bzeta^T \bP_o \bx \geq \left(1-\frac1{16}\right) (\bu^s)^T\bc - \varrho \bzeta^T \bqo \\
& \G(\bu^s,\bx) \tag{\laginnero}\\
& \btQ(\beta) = \{\bb^T\bx \leq \beta,\bPi \bx \leq \bqi,
\bx \in \bQ, \bx \geq 0\}
\end{align*}
\end{enumerate}
Then we can find a
$(1-(1 + a_1 + a_2)\epsilon)$-approximate solution to \ppone\ 
using $\tau$ rounds of deferred $\bu$-sparsifier construction, $\tau
=O\left(\rho_o\left(\frac1{\epsilon} + \epsilon \log \frac1{1-\epsilon_0}\right)\frac{\log (m/\epsilon)}{\log \promise} +  \frac{\log \ae}{\log \promise} \right)$.

 In each round we construct $O(\epsilon^{-1}\log \promise)$ deferred $\bu$--sparsifiers.
The algorithm invokes {\sc MicroOracle} for a fixed deferred sparsifier for at most $O(\rho_i (\log \rho_i) (\log \tno)\log \frac1\epsilon)$
times.
Each nonzero  $\bu_{\ell} \in  {\small \left[\left(2m/\epsilon\right)^{ - O\left(\frac{\rho_o}{\epsilon(1-\epsilon_0)} \right)},1 \right]}$ and finally $\log \L_0=O\left(\frac{\rho_o}{\epsilon(1-\epsilon_0)} \log (m/\epsilon) + \log \gamma \right)$. Moreover the algorithm computes 
(approximately) $\max \lambda; \bA\t{\bx} \geq \lambda \bc$ for a $\t{\bx}$ which is a weighted average of the $\bx$ returned by successive applications of part (ii) -- the algorithm can stop earlier than the stated number of $\tau$ rounds if it observes $\lambda \geq (1-3\epsilon)$.
\end{ntheorem}

\smallskip
We prove the theorem using the following two theorems; Theorem~\ref{useful1} is interesting in its own right and can be used when $\bPo = \bPi, \bqo=\bqi$ which implies $\rho_i=2$. However, the number of iterations depend on $\rho_o$ which can be reduced significantly while increasing $\rho_i$. Note that $\rho_i$ does not affect the overall number of iterations, but is present as a tool to construct the {\sc MicroOracle}.

\begin{theorem}
\label{useful1}{\rm (Proved in Section~\ref{proof:useful1}.)}
Suppose for a $(\rho_o,\rho_i)$--dual-primal amenable system \ppone\
given any $\bu^s \in \Re^{m}_+, \bu^s \geq 0$, a {\sc MiniOracle}  
provides either condition (i) as in Theorem~\ref{useful} or 
 a solution $\bx$ of \spar, where the number of non-zero entries of $\bx$ is at most $\tnx$.
{\small
\begin{align*}
& (\bu^s)^T \bA \bx \geq (1-\epsilon/8)(\bu^s)^T\bc \\
& \bPo \bx \leq 2\bqo \\
& \G(\bu^s,\bx)  \tag{\spar}\\
& \btQ(\beta) = \{\bb^T\bx \leq \beta,\bPi \bx \leq \bqi,
\bx \in \bQ, \bx \geq 0\}
\end{align*}
}
then the conclusions of Theorem~\ref{useful} (number of rounds, deferred sparsifiers constructed in each round, existence of early stopping certificate and the range 
of values for the multipliers) hold. 
\end{theorem}

\noindent The next theorem provides a surprising use of the sparsifiers, namely that we can iterate over the small size sparsifier without requiring fresh access to input.

\begin{theorem}{\rm (Proved in Section~\ref{proof:useful2}.)}
\label{useful2}
We use $\tau_i= O(\rho_i (\log \rho_i) (\log \tno)\log \frac1\epsilon)$ invocations of {\sc MicroOracle} to implement an 
invocation of \spar. 
 In combination with Theorem~\ref{useful1}, this implies Theorem~\ref{useful}. 
\end{theorem}

\subsubsection{Proof of Theorem~\ref{useful1}}
\label{proof:useful1}
We use the following result from \cite{PlotkinST95}, we change the notation to fit our context -- note that we start with a different initial condition than stated in Lemma~3.6 in \cite{PlotkinST95}. That lemma in \cite{PlotkinST95} relied on $O(M)$ rounds of computation to produce a $\epsilon_0=1-1/M$ solution, i.e., $\bA\bx_0 \geq \bc/M$ where $M$ is the number of constraints. The altered initial condition gives us Theorem~\ref{covering}.

\begin{theorem}[\cite{PlotkinST95}]
\label{covering}
Suppose we are given a decision
problem $\bA\bx \geq \bc$ such that $\bx \in \P$ where $\bA \in \R^{M
  \times N}, \bx \in \R^N, \bc \in \R^M$ where $\P$ is some polytope such that $\mathbf{0} \leq \bA \bx \leq \rho \bc$
for all $\bx \in \P$. 
Suppose we have an initial $\bx_0 \in \P$ satisfying $\bA\bx_0 \geq (1-\epsilon_0) \bc$. The algorithm sets $\bx=\bx_0$ and proceeds in phases. 
In phase $t$ it determines $\lambda_t=\min_{\ell} (\bA\bx)_{\ell}/\bc_\ell$. It then repeatedly 
queries an {\sc Oracle-C} for $\argmax_{\t{\bx} \in \P} \bu^T\bA \t{\bx}$ where 
\[ \bu_\ell = exp \left( - \alpha
(\bA \bx)_{\ell}/\bc_\ell \right)/\bc_\ell
 \qquad \alpha=O(\lambda^{-1}_t\epsilon^{-1} \ln (M/\epsilon)) \]
and performs an update step $\bx \leftarrow (1-\sigma)\bx+\sigma\t{\bx}$ where $\sigma=\epsilon/(4\alpha\rho)$ (note that the failure to produce a $\t{\bx}$ implies that the decision problem is infeasible). The phase continues till  $\lambda=\min_\ell (\bA\bx)_{\ell}/\bc_{\ell} \geq 
\max \{ 2\lambda_t,1 - 3\epsilon\}$. If $\lambda < 1-3\epsilon$, then a new phase is started till $\lambda \geq (1-3\epsilon)$.
Then after at most $T=O \left(\rho \left( \epsilon^{-2} + \log \frac{1}{1-\epsilon_0}\right) \log \frac{M}{\epsilon} \right)$ invocations of {\sc Oracle-C}  the fractional covering framework either provides (i) a solution $\bA\bx \geq (1-3\epsilon)\bc, \bx \in \P$ or (ii) a non-negative vector $\by \in \R^M$ such that $\by^T\bA\bx < \by^T \bc$ for all $\bx \in \P$, thereby proving the infeasibility of $\{\bA\bx \geq \bc , \bx \in \P\}$. 

The framework, explicitly  maintains $\lambda = \min_\ell (\bA\bx)_{\ell}/\bc_{\ell}$ and the $\bu_\ell$ change by a factor of at most $e^{\pm\epsilon}$ from one invocation of {\sc Oracle-C} to the next. Each nonzero  $\bu_{\ell} \in  \left[\left(2M/\epsilon\right)^{ - O\left(\frac{\rho_o}{\epsilon(1-\epsilon_0)} \right)},1 \right]$. 
\end{theorem}

\noindent We will be using a small modification which if not present in \cite{PlotkinST95}.

\begin{corollary}
\label{corcover}
Theorem~\ref{covering} holds if {\sc Oracle-C} finds $\bx \in \P$ satisfying
$\bu^T \bA \t{\bx} \geq (1-\epsilon/2) \bu^T \bc$
or reports that no such solutions exist.
\end{corollary}
\begin{proof}
Note that $\left(\bu^T \bA \t{\bx} \geq (1-\epsilon/2) \bu^T \bc \right)$ implies 
$\left( (1-\epsilon/2) \bu^T\bA\t{\bx} \geq (1+\epsilon/2) \bu^T\bA\bx + \epsilon \bu^T \bc/2 \right)$.
Lemma
  3.3 in \cite{PlotkinST95} requires the RHS of the implication to
  hold with $\epsilon \bu^T\bc$ (it uses notation $\by,\bb$ instead of $\bu,\bc$); however ensuring
  $\epsilon \bu^T\bc/2$ only increases the number of iterations in Theorem~\ref{covering} by a
  factor $2$. Lemma 3.1 in
  \cite{PlotkinST95} shows $\bu^T\bA\bx \leq (1+\epsilon)\lambda\bu^T\bc$ and therefore $(1+\epsilon/2) \bu^T\bA\bx \leq (1+\epsilon/2)(1+\epsilon)(1-3\epsilon) \bu^T\bb \leq (1-3\epsilon/2) \bu^T\bc$. The implication follows using simple calculations.
\end{proof}

\begin{algorithm}[H]
{\small
\begin{algorithmic}[1]
\vspace{0.05in}
\STATE Start with $\bx=\bx_0$ and $\beta=\beta_0$, where $\bx_0,\beta_0$ refer to the initial solution \initial.
\STATE Consider the following {\bf family} of decision problems:
{
\[ \hspace{-0.5in}
\left  .
\begin{array}{ll} 
& \{\bA\bx \geq \bc, \bx \in \P(\beta) \}  \quad \mbox{where} \\
&  \P(\beta)=  \left\{ \begin{array}{l}
\bP_o \bx \leq 2 \bqo  \\
\t{\bQ}(\beta) =  \left\{
\bb^T\bx \leq \beta;
\bP_i \bx \leq \bqi; \bx \in \bQ;
\bx \geq \bzero \right \}
\end{array} \right.  
\end{array} \right. \tag{\lpdual}
\]
}
\noindent We will be solving the entire family simultaneously. Note that $\t{\bQ}(\beta') \subseteq \t{\bQ}(\beta)$ for $\beta'  \leq \beta$. This is why we start from a low value of $\beta_0 \leq \beta^*/2$.
\WHILE {$\lambda < 1-3\epsilon$ (where $\lambda = \min_{\ell} (\bA\bx)_{\ell}/\bc_{\ell}$)}
\STATE Define $\bu$ as in Corollary~\ref{corcover}
and consider
the system \outero:
{\small
\[ \bu^T \bA \t{\bx} \geq (1-\epsilon/2) \bu^T \bc, \quad \t{\bx} \in \P(\beta) \tag{\outero}\] 
}
\vspace{-0.1in} 
\STATE Set $\bu^s=\bu$ and invoke {\sc MiniOracle} to solve \outero\
\vspace{-0.05in} 
\begin{itemize}\parskip=0in
\item[] (a) If the {\sc MiniOracle} returns $\t{\bx}$ then update $\bx,\bu$ as in Corollary~\ref{corcover}. 
\item[] (b) Otherwise (if {\sc MiniOracle} returns $\by$) then set $\beta \leftarrow (1+\epsilon) \beta$; remember only the last such $\by$; and repeat {\sc MiniOracle} with the new $\beta$.
\end{itemize}
\vspace{-0.05in}
\ENDWHILE
\STATE Output the last remembered $\by$ corresponding to $\beta'=\beta/(1+\epsilon)$. We prove that such a $\beta'$ exists.
\end{algorithmic}
}
\caption{An intermediate algorithm for Theorem~\ref{useful1}\label{alg:lessuseful1}} 
\end{algorithm}

\noindent 
\begin{proof}(Of Theorem~\ref{useful1}.)
As stated in the theorem, we have an initial solution
$\bx_0$ such that $\beta^*/\ae \leq \beta_0=\bb^T\bx_0 < \beta^*/2$ and 
$\bx \in \P(\beta_0)$ along with $\bA\bx_0 \geq (1
- \epsilon_0) \bc$. We consider the Algorithm~\ref{alg:lessuseful1} -- 
this is {\bf not} the final algorithm.

Based on Corollary~\ref{corcover}, using $M=m$ and $m \geq 1/\epsilon$,
if for some $\beta$ we have $T=O(\rho_o (\epsilon^{-2} + \log \frac1{(1-\epsilon_0)}) \log m)$ such solutions $\bx$; then we would have $\bA\bx
\geq (1-3\epsilon)\bc$ with $\bb^T \bx \leq \beta$. Since $\beta_0 \leq \beta^*/2$, it implies that at least one call to {\sc MiniOracle} will provide a $\by$ as in condition (i) since we are starting from $\beta=\beta^0$.
Now consider the {\em largest} value of $\beta$, say $\beta'$, for which any call to {\sc MiniOracle} has provided a $\by$ as in (i) -- such a value exists because we cannot have primal feasible solutions with value more than $\beta^*$.

Since an invocation to {\sc MiniOracle} has provided {\em primal feasible solution} $\by$ for \ppone\ satisfying $\bc^T \by \geq (1-a_2\epsilon)\beta'$. 
Now for the value of $(1+\epsilon)\beta'$ we have $T$ solutions of \outero, (irrespective of when and where we raised the $\beta$ values since a solution of \outero\ with $\beta=\beta''$, continues to hold
for larger values of $\beta > \beta''$), and therefore $\bA\bx
\geq (1-3\epsilon)\bc$ with $\bx \in Q,\bc^T\bx \leq (1+\epsilon)\beta'$.
By condition~(d\ref{dualfeasible}), we have a proof that $(1+\epsilon)\beta' \geq \beta^*/(1-a_2\epsilon)$.
Therefore we have a $\by$ (provided by the last invocation of {\sc MiniOracle} to return (i) ) such that $\bb^T\by \geq (1-(1+a_1+a_2)\epsilon)\beta^*$, proving the quality of approximation of Theorem~\ref{useful1} (same as in Theorem~\ref{useful}).
Note that the number of time we increased $\beta$ is at most $O(\frac1{\epsilon}\log \ae)$. Note that $\tau = (T + O(\frac1{\epsilon}\log \ae))/(\epsilon^{-1} \log \gamma)$, is the number of times we invoke {\sc MiniOracle}.
We modify Algorithm~\ref{alg:lessuseful1} to  Algorithm~\ref{alg:corealg} -- observe that the new algorithm simplifies to the description of Algorithm~\ref{alg:simplealg} 
discussed in Section~\ref{sec:general}.

\begin{algorithm}[H]
{\small
\begin{algorithmic}[1]
\vspace{0.05in}
\STATE Start with $\bx=\bx_0$ and $\beta=\beta_0$, where $\bx_0,\beta_0$ refer to the initial solution \initial.
\STATE Consider the following family of decision problems.
{\small
\[
\left . \hspace{-0.5in}
\begin{array}{ll} 
& \{\bA\bx \geq \bc, \bx \in \P(\beta) \}  \quad \mbox{where} \\
& \P(\beta)= \left\{ \begin{array}{l}
\bP_o \bx \leq 2 \bqo  \\
\t{\bQ}(\beta) = \left\{
\bb^T\bx \leq \beta;
\bP_i \bx \leq \bqi; \bx \in \bQ;
\bx \geq \bzero \right \}
\end{array} \right.  
\end{array} \right. \tag{\lpdual}
\]
}
\vspace{-0.05in}
\WHILE {$\lambda < 1-3\epsilon$ (where $\lambda=\min_\ell (\bA\bx)_{\ell}/\bc_{\ell}$)}
\STATE Define weights $\bu$ as in Corollary~\ref{corcover}. Compute $\frac{\ln \gamma}{\epsilon}$ deferred $\bu$-sparsifiers denoted by $\{\D_q\}$.
\FOR {$q=1$ to $\frac{\ln \gamma}{\epsilon}$}
\STATE Define \outero\ (based on the current $\bx,\bu(q)$) as:
{\small
\begin{align} \hspace{-0.5in} \bu^T(q) \bA \t{\bx} \geq (1-\epsilon/2) \bu^T(q) \bb \quad \t{\bx} \in \bQ_o(\beta) \tag{\outero}
\end{align} 
}
\vspace{-0.2in}
\STATE Refine $\D_q$ based on the current $\bx$ to produce $\bu^s(q)$.
\STATE Invoke {\sc MiniOracle} to solve \spar\
\vspace{-0.05in}
\begin{itemize}\parskip=0in
\item[] (a) If the {\sc MiniOracle} returns $\t{\bx}$; update $\bx$ as in Corollary~\ref{corcover}. Note $\spar\ \Rightarrow \outero$.
\item[] (b) Otherwise (if {\sc MiniOracle} returns $\by$) then set $\beta \leftarrow (1+\epsilon) \beta$; remember only the last such $\by$; and repeat {\sc MiniOracle} with the new $\beta$.
\end{itemize}
\vspace{-0.1in}
\ENDFOR
\ENDWHILE
\STATE Output the last remembered $\by$ corresponding to $\beta'=\beta/(1+\epsilon)$. We prove that such a $\beta'$ exists.
\end{algorithmic}
}
\caption{The algorithm for Theorem~\ref{useful1}\label{alg:corealg}} 
\end{algorithm}

We now observe that in each of the $T$ invocations of
{\sc MiniOracle} that provided a $\bx$, 
the $\bu_\ell$ change by a factor of $e^{\pm
\epsilon}$. Therefore if we perform $\epsilon^{-1}\log \gamma$
invocations the values of $\bu_\ell$ (for every $\ell$) will change by a 
factor in the range $[1/\gamma,\gamma]$. But then we can apply deferred
$\bu$-sparsifiers. Condition~(\ref{defcon}) in
Definition~\ref{def:amenable} asserts that \spar\ (applied to $\bu^s$)
still implies \outero\ (applied to $\bu$), and the above proof of
quality of approximation and number of invocations remain valid. The
Theorem now follows from observing that we can construct $e^{\pm
\epsilon}$ deferred $\bu$-sparsifiers independently and in
parallel. We use the results of the $q^{th}$ invocation (which
contains the results of all invocations $1,\ldots,q$) to refine the
$(q+1)^{st}$ deferred sparsifier and then use that refinement in the
$(q+1)^{st}$ invocation. The theorem follows -- observe that the
guarantee on early stopping is provided by Corollary \ref{corcover}.
Note that the bound of $\bu_\ell$ (statement of Theorem~\ref{useful}, also holds in Theorem~\ref{useful1}) follows from simple inspection. Note $\L_0$ is within a factor of $1/\gamma$ of $\min_{\ell} \bu_\ell$ and the result follows.
\end{proof}

\subsubsection{Proof of Theorem~\ref{useful2}}
\label{proof:useful2}

We rewrite \spar\ as follows:
{\small
\begin{align*}
& \bPo \bx \leq \Lambda \bqo  \qquad \Lambda=2\\
\bQ(\bu^s,\beta): & \left \{ \begin{array}{ll}
& (\bu^s)^T \bA \bx \geq (1-\epsilon/8)(\bu^s)^T\bc \\
& \G(\bu^s,\bx)  \tag{\spar}\\
& \btQ(\beta) = \{\bb^T\bx \leq \beta,\bPi \bx \leq \bqi,
\bx \in \bQ, \bx \geq 0\}
\end{array} \right.
\end{align*}
}
\noindent \spar\ defines a packing problem. We now consider the following theorem in \cite{PlotkinST95} --
that paper used the notation $\bA,\lambda,\bb,\P,\bu,\epsilon,\rho,\alpha,\sigma$ instead of $\bA^p,\lambda^p,\bd,\P^p,\bz,\delta,\rho',\alpha',\sigma'$ respectively -- we use different notation since several of those symbols will be in use when this theorem is applied. 

\begin{theorem}[\cite{PlotkinST95}]
\label{packing}
Suppose we are given a decision
problem $\bA^p\bx \leq \bd$ such that $\bx \in \P^p$ where $\bA^p \in \R^{M'
  \times N'}, \bx \in \R^{N'}, \bd \in \R^{M'}$ where $\P^p$ is some polytope such that $\mathbf{0} \leq \bA^p \bx \leq \rho' \bd$
for all $\bx \in \P^p$. 
Suppose we have an initial $\bx_0 \in \P^p$ satisfying $\bA^p\bx_0 \leq \delta_0 \bd$. The algorithm sets $\bx=\bx_0$ and proceeds in phases. 
In phase $t$ it determines $\lambda^p_t=\max_r (\bA^p\bx)_r /\bd_r$. It then repeatedly 
queries an {\sc Oracle-P} for $\argmin_{\t{\bx} \in \P^p} \bz^T\bA^p \t{\bx}$ where 
\[\bz_r = exp( \alpha'
(\bA^p \bx)_r/\bd_r)/\bd_r
 \qquad \alpha'=O((\lambda^{p}_t)^{-1}\delta^{-1} \ln (M'/\delta))\]
and performs an update step $\bx \leftarrow (1-\sigma')\bx+\sigma'\t{\bx}$ where $\sigma'=\delta/(4\alpha'\rho')$.
The phase continues till  $\lambda^p=\max_r (\bA^p\bx)_r/\bd_r \leq 
\min \{ \lambda^p_t/2,1 + 6\delta\}$. If $\lambda^p > 1+6\delta$, then a new phase is started till $\lambda^p \leq (1+6\delta)$.
Then after at most $T=O(\rho' (\delta^{-2} +\log \delta_0)\log M')$ successful invocations of {\sc Oracle-P}  the fractional packing framework provides a solution $\bA^p\bx \leq (1+6\delta)\bd, \bx \in \P^p$. 
\end{theorem}

\noindent 
 We will again be using a modification which is not present in \cite{PlotkinST95}.
The proof of Corollary~\ref{corpacking} follows from observing that
$\bz^T \bA^p \t{\bx} \leq (1+\delta/2) \bz^T \bd$ implies
$(1+\delta/2) \bz^T\bA^p\t{\bx} \leq (1-\delta/2) \bz^T\bA^p\bx - \delta \bz^T \bd/2$ 
and then using exactly the same arguments as was used in the proof of Corollary~\ref{corcover} modifying Theorem~\ref{covering}.

\begin{corollary}
\label{corpacking}
Theorem~\ref{covering} holds if {\sc Oracle-P} finds $\bx \in \P^p$ satisfying
$\bz^T \bA^p \t{\bx} \leq (1+\delta/2) \bz^T \bb$.
\end{corollary}

We will use Corollary~\ref{corpacking} with $\delta=\frac16$ and $\P^p=\bQ(\bu^s,\beta)$ on the packing problem Modified-\spar. The solution desired by {\sc Oracle-P} is given by \innero, and $\rho'=\rho_i$ from Definition~\ref{def:amenable}.
{\small
\[
\begin{minipage}[c]{0.5\textwidth}
\begin{align*}
&\mbox{(Modified-\spar)}\\
& \bPo \bx \leq \bqo  \\
\bQ(\bu^s,\beta): & \left \{ \begin{array}{ll}
& (\bu^s)^T \bA \bx \geq (1-\epsilon/8)(\bu^s)^T\bc \\
& \G(\bu^s,\bx)  \\
& \btQ(\beta) 
\end{array} \right.
\end{align*}
\end{minipage}\quad\left|
\hspace*{-0.5cm}
\begin{minipage}[c]{0.50\textwidth}
\begin{align*}
&\mbox{(\innero)}\\
& \bz^T \bP_o \bx \leq (13/12) \bz^T \bqo \\
\bQ(\bu^s,\beta): & \left\{ \begin{array}{ll}
& (\bu^s)^T \bA \bx \geq \left(1-\frac\epsilon8\right) (\bu^s)^T\bc \\
&  \G(\bu^s,\bx)  \\
& \btQ(\beta) 
\end{array}\right.
\end{align*} 
\end{minipage} \right.
\]
}
\begin{corollary}
\label{trcor}
If we have a {\sc TransientOracle} that either provides a $\by$ as desired by condition (i) in Theorem~\ref{useful} or provides a solution of \innero\ for $O(\rho_i (\log \rho_i) \log \tno)$ steps for the $\bz$ specified by Corollary~\ref{corpacking}, the Theorem~\ref{useful} holds.
\end{corollary}
\begin{proof}
Set $\delta=\frac16$.
Observe that if any of the invocations to {\sc TransientOracle} provided a $\by$ we implement part (i) of the {\sc MiniOracle}. If all the invocations succeed then we have a solution of \spar. The result follows from observing $\delta_0=\rho'$ and $M'=\tno$. Note that any solution of \innero\ provides an initial solution with $\delta_0 = \rho_i$ for the application of Corollary~\ref{corpacking}.  
\end{proof}

We will not be concerned with unsuccessful invocations of {\sc Oracle-P}
-- the final Lagrangian used in the desiderata of {\sc MicroOracle} is stronger by a factor of $2$; i.e., even though we want to satisfy $\bP_o \bx \leq 2\bqo$ we have constraint $LHS \geq (\bu^s)^T\bc - \varrho \bz^T \bqo$ which is much stronger than $LHS \geq (\bu^s)^T\bc - \varrho 2\bz^T \bqo$. Theorem~\ref{useful2} now follows from the following lemma and Corollary~\ref{trcor}:
\begin{lemma}
If we have a {\sc MicroOracle} that provides either a $\by$ as desired by condition (i) in Theorem~\ref{useful} or a solution of \laginnero\ for any $0 < \varrho$, then Theorem~\ref{useful1} holds.
{\small
\[\hspace*{-0.5cm}
\begin{minipage}[c]{0.45\textwidth}
\begin{align*}
& \bz^T \bP_o \bx \leq (13/12) \bz^T \bqo \\
\P: & \left\{ \begin{array}{ll}
& (\bu^s)^T \bA \bx \geq \left(1-\frac\epsilon8\right) (\bu^s)^T\bc \\
& \G(\bu^s,\bx) \\
& \btQ(\beta) 
\end{array}\right . \hspace{-0.5cm} \tag{\innero}
\end{align*}
\end{minipage}
\left | \hspace{-1cm}
\begin{minipage}[c]{0.65\textwidth}
\begin{align*}
& (\bu^s)^T \bA \bx - \varrho \bzeta^T \bP_o \bx \geq \left(1-\frac\epsilon{16}\right)\left[ (\bu^s)^T\bc - \varrho \bzeta^T \bqo \right]\\
& \G(\bu^s,\bx) \tag{\laginnero}\\
& \btQ(\beta) = \{\bb^T\bx \leq \beta,\bPi \bx \leq \bqi,
\bx \in \bQ, \bx \geq 0\}
\end{align*}
\end{minipage}\right. 
\]
} 
\end{lemma}
\begin{proof}
Note that if any invocation returns a $\by$ then we have nothing to prove. Therefore we focus on the solutions to \laginnero.
Since $\bz^T\bqo \geq 0$, a solution of \laginnero\ also implies:
{\small
\begin{align}
& (\bu^s)^T \bA \bx - \varrho \bz^T \bP_o \bx \geq \left(1-\frac\epsilon{16}\right) \left[(\bu^s)^T\bc - \varrho \bz^T \bqo \right]  \geq \left(1-\frac\epsilon{16}\right) \left[(\bu^s)^T\bc - \frac{13}{12}\varrho \bz^T \bqo \right]
\label{abceqn}
\end{align}
}
\noindent We invoke the {\sc MicroOracle} with $\varrho=\epsilon (\bu^s)^T\bc/(16\bzeta^T\bqo)$. If the returned solution $\bx$ satisfies 
{\small
\begin{align} 
\bz^T \bP_o \bx \leq (13/12) \bz^T \bqo \label{eqn0001}
\end{align}
} 
then we have immediately a solution for \innero\ since 
{\small
\begin{align*}
 (\bu^s)^T \bA \bx - \varrho \bz^T \bP_o \bx  & \geq \left(1-\frac\epsilon{16}\right) \left[(\bu^s)^T\bc - \varrho \bz^T \bqo \right] = \left(1-\frac\epsilon{16}\right)^2(\bu^s)^T\bc   > \left(1-\frac\epsilon{8}\right)(\bu^s)^T\bc
\end{align*}
}
Now if $\varrho \geq \varrho_0 = \frac{12\bu_s^T\bc}{13\bz^T \bqo}$ then $\bx=\bzero$ is a feasible solution for 
Equation~\ref{abceqn} since right hand size is $0$. Note $\bx=\bzero$ also satisfies $\G(\bu^x,\bx),\bQ$ by Definition~\ref{def:amenable} and definitely satisfies Equation~\ref{eqn0001}.
Therefore we can perform a binary search over $\varrho$ and finally 
achieve an interval $[\varrho_1,\varrho_2]$ where in the corresponding solutions; $\t{\bx}_1$ does not satisfy Equation~\ref{eqn0001} and  $\t{\bx}_2$ does (but both satisfy 
Equation~\ref{abceqn}, $\G(\bu^s,\bx)$ and $\btQ(\beta)$).
Let $\Upsilon = \frac{13}{12} \bz^T \bqo$.
Eventually we can narrow the interval $\varrho_2-\varrho_1 \leq \epsilon \varrho_0/16$, where $\varrho_2>\varrho_1$ and
\begin{align*} 
\bz^T \bP_o \t{\bx}_1 = \Upsilon_1 > \Upsilon = \frac{13}{12} \bz^T \bqo \mbox{~and~}
\bz^T \bP_o \t{\bx}_2 = \Upsilon_2 \leq \Upsilon
\end{align*} 
We then find two numbers $s_1,s_2$ such that $s_1 + s_2 =1$, $s_1\Upsilon_1+ s_2\Upsilon_2= \Upsilon = \frac{13}{12} \bz^T \bqo$. Let $\bx=s_1\t{\bx}_1+s_2\t{\bx}_2$. Observe that $\bx$ satisfies Equation~\ref{eqn0001}, $\G$ and $\btQ(\beta)$. Since $\t{\bx}_1,\t{\bx}_2$ both satisfy Equation~\ref{abceqn},

{\small
\begin{align*}
 (\bu^s)^T \bA \bx  & = s_1 (\bu^s)^T \bA \t{\bx}_1 +
s_2 (\bu^s)^T \bA \t{\bx}_2 \\ 
& \geq s_1\left( \left(1-\frac\epsilon{16}\right) \left[ (\bu^s)^T\bc - \varrho_1\frac{13}{12} \bz^T\bqo\right] + \varrho_1 \bz^T \bP_o \t{\bx}_1 \right)   + s_2\left( \left(1-\frac\epsilon{16}\right)\left[(\bu^s)^T\bc - \varrho_2 \frac{13}{12} \bz^T\bqo\right] + \varrho_2 \bz^T \bP_o \t{\bx}_2 \right) \\
& = \left(1-\frac\epsilon{16}\right) \left[(\bu^s)^T\bc - s_1 \varrho_1 \Upsilon - s_2 \varrho_2 \Upsilon\right]  +
s_1\varrho_1\Upsilon_1 + s_2\varrho_2\Upsilon_2\\
& = \left(1-\frac\epsilon{16}\right) \left[(\bu^s)^T\bc - \varrho_1 \Upsilon \right] -
\left(1-\frac\epsilon{16}\right)(\varrho_2-\varrho_1)s_2 +
\varrho_1 \Upsilon + s_2(\varrho_2-\varrho_1)\Upsilon_2 \\
& \geq \left(1-\frac\epsilon{16}\right) (\bu^s)^T\bc - s_2(\varrho_2 - \varrho_1)\Upsilon
\geq \left(1-\frac\epsilon{16}\right) (\bu^s)^T\bc - \frac{\epsilon\varrho_0}{16}\Upsilon \\
& \geq \left(1-\frac\epsilon{16}\right) (\bu^s)^T\bc - \frac\epsilon{16} (\bu^s)^T\bc
\end{align*}
}
This completes the proof of the lemma and of Theorem~\ref{useful2} (and of Theorems~\ref{useful1} and \ref{useful} as well).
\end{proof}

\subsection{Connections to Mirror Descent and Prox}
\label{mirror}
Mirror 
Descent algorithms were invented by Nemirovski and Yudin \cite{NY83}, 
see also \cite{nem,BT03,nest09}.  These were shown to be useful in the
context of solving large scale convex optimization problems.  In this
setting, the overall ``primal-dual'' algorithm in every step (i) makes
a projection to the dual space, (ii) constructs an update in the ``dual''
space and (iii) projects the update back in the original space.  
While these types of algorithms use duality -- {\em they use
the Fenchel Duality, and not the Lagrangian Duality}. 
The
dual space is defined by the {\em Legendre-Fenchel Transformation},
i.e., taking the convex-conjugate of a convex function $f^*(x^*)= \sup
\{ x^Tx^* - f(x) \}$
and $(f^*)^* =
f$. Fenchel Duality allows us to define distances between convex spaces
and uses the Bregman projections to transform between these spaces
\cite{onlinenotes}. Intuitively 
these notions capture how the convergence is measured (the loss function).
Measurement metrics however are only one aspect of the structure in an
optimization problem. We are interested in {\em representations},
namely constraints and the structure implied by them, and therefore the
explicit representation of the dual polytope provided by the
Lagrangians is more important to our context.  Of course, these different
notions of duality are related -- one can even view the
Fenchel-Duality as removing the auxiliary (or unimportant) dual
variables by taking the supremum. 
While the ``spirit'' of the Mirror descent algorithms,
an iterative algorithm using updates in a dual space, is the same as
in our context; the substance in these two contexts are less related.

\smallskip
The first order methods such as Prox \cite{prox1,prox2} provide $O(1/t)$ guarantee 
on error in $t$ rounds, but their intended use cases are different. They seek to solve a convex minimization, (using notation of \cite{prox1})
where $Q_1,Q_2$ are bounded convex spaces and $\h{f},\h{\phi}$ are convex functions:

{\small
\[ \min_{\bx \in Q_1} f(\bx) = \min_{\bx \in Q_1} \left( \h{f}(\bx) + \max_{\bu \in Q_2}  \left( \bu^T\bA\bx - \h{\phi}(\bu) \right) \right) \]
}

In both, at time $t$ the candidate solution $\bx_t$ satisfies $f(\bx_t) - \min_{\bx} f(\bx) = O(K/t)$.
Even though there are many equivalent representations of the same $f()$; this result is not scale free and depends on the 
value of the maximum matching. Moreover for a weighted problem, the boundedness assumption of $Q_1,Q_2$ are altered 
as the weights are scaled. These results do not give good approximations to weighted matching.
For unweighted matching, the $K$ in \cite{prox1} is large $\approx \|\bA\|_{1,2}$  where 

{\small
\begin{equation}
 \|\bA\|_{1,2} = \max_{\bx,\bu} \{\langle \bA\bx, \bu \rangle_{2} : \|\bx\|_1=1, \|\bu\|_2=1 \} \label{bounda}
\end{equation}
}

If we were to use the standard form of maximum matching (even over small sets) we will have $\bA\bx \geq \bc$ corresponding to 
$x_i + x_j + \sum_{U \in \O_s} z_U \geq w_{ij}$.
Now $\bA$ is a $0/1$-matrix and thus the bound in Equation~\ref{bounda} is at least the minimum
of number of $1$'s in a row or a column. If the maximum degree is $d$ then a vertex 
has $d$ edges and some edge appears in at least $d$ odd-set constraints of size $3$. However the maximum matching can be as small as $O(n/d)$; which implies $O(d^2/(n\epsilon))$ steps to achieve a $(1+\epsilon)$ multiplicative approximation and this is not independent of $n$ for $d \gg \sqrt{n}$.

\cite{prox2} combines \cite{prox1} and mirror descent and shows that if suitable projections exist, 
then the parameter $K$ can be the Lipshitz parameter associated with the projection.
It is not clear what that projection should be for a particular problem. Moreover for a non-bipartite graph, the projections seems to have a size $\min\{m,n^{O(1/\epsilon)}\}$. If we are allowed random access to $O(m)$ space, then the number of rounds is $1$, since all edges can be stored.

\section{Weighted Nonbipartite $b$--Matching.} 
\label{sec:ex}

In this section we show the application of Theorem~\ref{useful} in the
context of $b$--Matching. Ideally, that demonstration should have
shown three separate examples -- the unweighted nonbipartite case,
weighted bipartite case and finally the weighted nonbipartite
case. Although the last case is more general, a direct proof of the
weighted bipartite case would avoid some of the complexities of the
weighted nonbipartite case. In the interest of space we provide a
single proof. We start with:

\begin{definition}
\label{defemax}
Let $W^*=\max_{(i,j) \in E} w_{ij}$. For $k \geq 0$, let $\h{w}_k=\left(1+ \epsilon\right)^{k}$. Using $O(p)$ rounds and $n^{1+1/p}$ space we can easily find an edge with the maximum
weight $W^*$ (using $\ell_0$ sampling, which can be implemented using sketches). 
\end{definition} 
\begin{definition}
\label{defek}
Given an edge
$(i,j) \in E$ we can define an unique level $k \geq 0$ such that
$\frac{\epsilon W^*}{B}\h{w}_k \leq w_{ij}
< \frac{\epsilon W^*}{B}\h{w}_{k+1}$ and let $\h{w}_{ij}$ be that $\h{w}_k$. 
Let $\h{E}_k = \big \{ (i,j) \big | (i,j) \in E, \h{w}_{ij} =\h{w}_k = \left(1 +\epsilon \right)^k \big\}$ and let $\h{E} = \bigcup_k \h{E}_k$. Let $\h{w}_L$ correspond to the largest weight class in $\h{E}$. Observe that $L=O(\frac1\epsilon \ln B)$ and the number of levels is $L+1$. 
\end{definition}

\noindent Recall the system \ref{lpbm} that defines the maximum weighted non-bipartite $b$--Matching. We will however consider the alternate rescaled system \ref{nicerlp0}.

\begin{observation}
\label{tlemma}
$\h{\beta} \leq \frac{B}{\epsilon W^*} \beta^* \leq \frac{(1+\epsilon)}{(1-\epsilon)} \h{\beta}$ where $\h{\beta}$ is defined in \ref{nicerlp0}. Moreover any $b$--Matching in $\h{E}$ corresponds to a $b$--Matching in $E$ under rescaling of edge weights.

{\small
\begin{align*}
& \h{\beta}=\max \sum_{(i,j) \in \h{E}} \h{w}_{ij} y_{ij} \\
&\sum_{j:(i,j) \in \h{E}} y_{ij} \leq b_i  & \forall i \lptag \label{nicerlp0}\\
&\displaystyle \sum_{(i,j) \in \h{E}:i,j \in U} y_{ij} \leq \left \lfloor \bnorm{U}/{2} \right \rfloor  & \forall U \in \O \\
& y_{ij} \geq 0 \qquad & \forall (i,j) \in \h{E}
\end{align*}
}
\end{observation}

\eat{
Observation~\ref{tlemma} follows from $W^* \leq \beta^*$, $\sum_{(i,j) \in E} w_{ij} y_{ij} \leq BW^*$ and that the constraints of \ref{nicerlp0} imply the constraints of \ref{lpbm}.  Further a feasible solution of \ref{lpbm}, removing the edges the edges in $E - \h{E}$ gives a feasible solution for \ref{nicerlp0} with the desired bound.
}
In the remainder of the discussion we find an integral $(1-O(\epsilon))$ -approximate solution of \ref{nicerlp0} and use the observation to show that the same solution is an integral $(1-O(\epsilon))$ -approximate solution of \ref{lpbm}.
Moreover the transformation of $w_{ij} \rightarrow \h{w}_k$ in can be achieved efficiently.
 We define $\bA,\bc,\bPo,\bPi,\bqo,\bqi$ as follows, $\bQ =\{x_i - x_{i(\ell)} \geq 0, \forall i,\ell\}$. We define $\{\bA\bx \geq \bc\}$  as

{\small $$ x_{i(k)} + x_{j(k)} + \sum_{\ell \leq k} \left( \sum_{U\in \O_s,i,j \in U} z_{U,\ell} \right) \geq \h{w}_{k} \qquad \forall (i,j) \in \h{E}_k $$}
and $\{\bPo\bx\leq \bqo\}$ as
{\small $$ 2x_{i(k)} + \sum_{\ell \leq k} \left( \sum_{U \in \O_s:i, \in U} z_{U \ell} \right) \leq 3\h{w}_k \qquad \forall i,k $$}

We define $\{\bPi\bx\leq \bqi\}$ as
{\small $$ 2x_{i(k)}+\sum_{\ell \leq k} \left( \sum_{U \in \O_s:i, \in U} z_{U \ell} \right) \leq 
\left(\frac{24}{\epsilon} + \frac{24}{\epsilon^2} \right)
\h{w}_k \qquad \forall i,k $$}
And finally we define $\G(\bu^s,\bx)$ as for all $U \in \O_s,\ell$
{\small
\begin{align*} 
& z_{U,\ell} \left( \sum_{k \geq \ell } \left( \sum_{(i,j) \in \h{E}_k, i,j \in U} \hspace{-0.5cm} u^s_{ijk}  -  \sum_{i \in U}  \left( \sum_{j \not \in U, (i,j) \in \h{E}_k} u^s_{ijk} \right) \right) \right) \geq 0 
\end{align*}
}

We show that \ref{nicerlp0} is $ \left(6,O(\epsilon^{-2}) \right)$--dual-primal amenable with $a_1=
3$, since a solution of $\bA\bx \geq (1-3\epsilon)\bc$, setting $x_i = \frac1{1-3\epsilon}\max_\ell x_{i(\ell)}, z_U = \frac1{1-3\epsilon}\sum_{\ell} z_{U,\ell}$ will satisfy the dual of \ref{nicerlp0}. Therefore the
conditions~(d\ref{dualfeasible})--(d\ref{innerwidth}) in
Definition~\ref{def:amenable} hold with $a_1=3$.  Note that for $\bPo\bx \leq \bqo$ we do not need to have constraints for $i,k$ if the vertex $i$ has no edges of level $k$ incident to it. The same holds for $\bPi\bx \leq \bqi$. Therefore $\tno=\tni=\min \{m,\frac{n}\epsilon \log B\}$.
We focus on 
condition~(d\ref{defcon}) in Definition~\ref{def:amenable}. Let $\bx=\{x_{i(k)},x_i,z_{U,\ell}\}$. Note $u_{ijk}$ is: 

\vspace{-0.1in}
{\small
\[ \frac{1}{\h{w}_{k}} \exp\left( - \alpha \left ( x_{i(k)} + x_{j(k)} + \sum_{\ell \leq k} \left( \sum_{U\in \O_s,i,j \in U} z_{U,\ell} \right) \right)  /\h{w}_{k} \right)
\]
}
and with $\bu=\{u_{ijk}\}$, we have:
{\small
\begin{align*}
\bu^T\bA\bx & =\sum_{i,k} x_{i(k)} \left( \sum_{j: (i,j) \in \h{E}_k} u_{ijk} \right)   +
\sum_{U \in \O_s,\ell} z_{U,\ell} \left( \sum_{k \geq \ell} \left( \sum_{(i,j) \in \h{E}_k,i,j \in U} 
u_{ijk} \right) \right)
\end{align*}
}

The next lemma follows from 
sparsifiers.  Note that the
sparsifiers are computed separately for each class of edges, and union
of sparsifiers constructed for each class is a sparsifier itself over
the entire set of edges. For each weight class $k$ we construct a
``weighted'' sparsifier where the ``weight'' of an edge $(i,j)$ is given by
$u_{ijk}$. 

\vspace{-0.05in}
\begin{lemma}\label{gluelemma}
If for each $k\geq 0$ we have $H_k=(V,E'_k, \{u^s_{ijk}\})$ as a
$(1\pm\epsilon/16)$-Cut-Sparsifier for $G^w_k=(V,\h{E}_k,\{u_{ijk}\})$, 
then let their union be $\bu^s=\{u^s_{ijk}\}$. Then, 

\vspace{-0.1in}
{\small
\begin{align*}
& \displaystyle (\bu^s)^T \bA\bx \geq \sum_k \left( \sum_{(i,j)\in \h{E}_k} \h{w}_k u^s_{ijk} \right), \ \G(\bu^s,\bx), \ \bx \in \bQ \quad \implies \quad  \bu^T \bA\bx \geq \left(1-\frac{\epsilon}{2}\right) \sum_k \left( \sum_{(i,j) \in \h{E}_k} \h{w}_k u_{ijk} \right) 
\end{align*}
}
which is the desired equation \simp\ in Condition~(d\ref{defcon}) of Definition~\ref{def:amenable}. 
\end{lemma}
\noindent
For
condition~(d\ref{definit}) observe that setting all $z_{U,\ell}=0$
corresponds to a bipartite relaxation. Let $\beta^b$ be the optimum value of the bipartite relaxation, $\h{\beta} \leq \beta^b \leq \frac32 \h{\beta}$.  We use:

\begin{lemma}\label{init-weighted-bip}{\rm 
(Initial Solution)
}
Given $\{M_k\}$, where each $M_k$ is a maximal $b$--Matching for $\h{E}_k$, we  construct an initial 
solution $\bx_0=\{x_i\},\{x_{i(k)}\}$ satisfying $\bQ,\bPo$, $\bA\bx_0 \geq (1-\epsilon_0)\bc$ and $\frac{\beta^b}{\ae} \leq \beta_0=\bb^T\bx_0 = \sum_i b_ix_i \leq \frac{\beta^b}{2}$ where $\ae=2048 \epsilon^{-2}$ and $\epsilon_0 = 1 - \epsilon/256$. $\{M_k\}$ is computed 
using $n_{init}= O(n^{1+1/(2p)}L)$ space and $O(p)$ rounds of sketching.
\end{lemma}

\noindent The {\sc MicroOracle} is provided by the next two lemmas:

\begin{lemma}\label{oolemma}{\rm (Part (i) of the {\sc MicroOracle})} 
For any $0<\epsilon \leq \frac1{16}$, suppose we are given a subgraph $G=(V,E')$
where $|V|=n$ and the weight  $\h{w}_{ij}$ of every edge $(i,j) \in E$ is of the form $\h{w}_k=(1+\epsilon)^k$ for $k\geq 0$, and thus $E'_k=\{(i,j)| \h{w}_{ij} =\h{w}_k\}$ and $E' = \cup_k E'_k$.
If we are also given a 
feasible solution to the system \ref{nicerlp} 
then we find an integral solution of \ref{nicerlp0} of weight $(1-2\epsilon)\beta$ 
using $O(|E'| \poly (\epsilon^{-1}, \log n))$ time using edges of $E'$.

{\small
\begin{align*}
& \sum_{k} \h{w}_k \left( \sum_{(i,j) \in E'_k} y_{ij} - 3\sum_{i}  \mu_{ik} \right) \geq (1-\epsilon)\beta  & \lptag \label{nicerlp} \\
&\sum_{j:(i,j) \in \h{E'}_k}\left( y_{ij} - 2 \mu_{ik}\right)\leq y_{i(k)}  & \forall i,k \\
&\sum_k y_{i(k)} \leq b_i & \forall i \\
&\displaystyle  \sum_{k \geq \ell} \left( \sum_{(i,j) \in \h{E}_k:i,j \in U}  y_{ij} - \sum_{i \in U} \mu_{ik} \right)\leq \left \lfloor \frac{\bnorm{U}}{2} \right \rfloor  & \forall U \in \O_s, \forall \ell  \\
& y_{ij},y_{i(k)},\mu_{ik} \geq 0 & \forall (i,j) \in \h{E}, i,k
\end{align*}
}
\end{lemma}

\begin{lemma}{\rm (Part (ii) of The {\sc MicroOracle}.)} \label{weighted-oracle}
Suppose we are given nonnegative $\{u^s_{ijk}\},\{\zeta_{ik}\},\beta$ and $\epsilon \in (0, \frac1{16})$, $\varrho>0$, such that $u^s_{ijk}$ corresponds to an edge $(i,j)$ and $k\geq 0$. Suppose further that there exists at most one $k$ such that $u^s_{ijk} \neq 0$. Let $E'_k = \{ (i,j) | u^s_{ijk} \neq 0\}$ and $E' = \cup_k E'_k$. Then using time $O(|E'|\poly(\epsilon^{-1},\log n))$  we either provide (i) a solution to the system \ref{nicerlp}
or (ii) a solution to the system \ref{bwm2} along with $\G(\bu^s,\bx)$ where $\h{w}_k=(1+\epsilon)^k$.

{\small
\begin{align*}
& \displaystyle \sum_{i,k} x_{i(k)} \left( \sum_{j: (i,j) \in E'_k} u^s_{ijk} - 2\varrho\zeta_{ik}\right) + \sum_{U \in \O_s,\ell} z_{U,\ell} \left( \sum_{k \geq \ell} \left( \sum_{(i,j) \in E'_k,i,j \in U} 
u^s_{ijk} - \varrho \sum_{i \in U} \zeta_{ik} \right) \right)  \\
& \hspace{2in} \geq \left(1-\frac{\epsilon}{16}\right) \sum_k \h{w}_k \left( \sum_{(i,j) \in E'_k} u^s_{ijk}  - 3\varrho  \sum_i \zeta_{ik} \right) \\
& \sum_{k \geq \ell} \left( \sum_{(i,j) \in \h{E}_k, i,j \in U} u^s_{ijk}  -  \sum_{i \in U} \left( \sum_{j \not \in U, (i,j) \in \h{E}_k} u^s_{ijk} \right) \right)  \geq 0  
& \forall U \in \O_s,\ell \\
& \sum_i b_i x_{i} + \sum_{\ell,U \in \O_s} z_{U,\ell} \left\lfloor \bnorm{U}/2 \right\rfloor \leq \beta  \lptag \label{bwm2} \\
& \displaystyle  x_{i(k)} + \sum_{U \in \O_s:i \in U} \sum_{\ell \leq k}
z_{U, \ell} \leq \left(\frac{24}{\epsilon} + \frac{24}{\epsilon^2} \right) \h{w}_k &  \forall  i,k \\
& x_i - x_{i(\ell)} \hspace{2in} & \forall i,\ell\\
&  x_{i(k)},x_i,z_{U,k} \geq 0 \hspace{1.6in} &\forall  i,U,k   
\end{align*}
}
Furthermore $\{U|z_{U,\ell} >0\}$ are disjoint for any fixed $\ell$ which shows $\tnx=O(nL)$.
\end{lemma}

The proof of Lemma~\ref{oolemma} is based on showing that 
if we restrict ourselves to the edges $\{(i,j)| y_{ij} >0\}$, then given an optimal 
dual solution of \ref{nicerlp0} (where the first constraint expressed as a maximization), we can produce a feasible solution to the dual of 
\ref{nicerlp} which is less than $(1+\epsilon)$ times the optimal 
dual solution of \ref{nicerlp0}. Thus we have a lower bound to the feasible solution of the dual of 
\ref{nicerlp}, and in turn a lower bound to the optimal 
dual solution of \ref{nicerlp0}. But that lower bound implies that
there exists a large primal solution restricted to those edges. 
Lemma~\ref{weighted-oracle} either provides a solution of \ref{nicerlp}, which using Lemma~\ref{oolemma} proves the existence of a large matching (as a factor of $\beta$), or Lemma~\ref{weighted-oracle} makes progress towards solving the dual (proving that the current bound of $\beta$ is appropriate). 
We use any offline algorithm to construct a $(1-\epsilon)$ approximation to the maximum matching on the union of the edges stored by the $t$ deferred sparsifiers (e.g., \cite{DuanP10,AhnG14}) and adjust $\beta$ and the multipliers as in Algorithm~\ref{alg:simplealg}.
 We summarize as follows:

\begin{theorem}
\label{thm3}
For any constant $0<\epsilon\leq \frac1{16}$ and $p>1$,
we can find a $(1-14\epsilon)$ approximate integral weighted $b$--Matching for nonbipartite graphs in $O(m \poly(\epsilon^{-1},\log n))$ time, $O(p/\epsilon)$ rounds and $O(n^{1+1/p}\log B)$ centralized space.
\end{theorem}

\subsection{New Relaxations and Oracles: Lemma~\ref{weighted-oracle}}
\label{sec:ben}

\vspace{0.1in}
The micro-oracle is provided by Algorithm~\ref{alg:micro3}. Observe that if
$\gamma \leq 0$ then the Lemma is trivially satisfied by setting
$x_i=x_{i(\ell)}=z_{U,\ell}=0$ for all $i,U,\ell$.
 The overall idea of the algorithm is simple: {\em find the violated constraints of \ref{nicerlp0}} -- either the contribution of those constraints 
are small in comparison to a rescaled value of $\beta$, in which case simply zeroing out the associated variables (carefully) will preserve a large solution and satisfy \ref{nicerlp0} -- or the contribution is large and thus giving small weights to each of them we satisfy \ref{bwm2}.
The algorithm~\ref{alg:micro3} uses the following lemma: 

\begin{lemma}
\label{weighted-oracle-helper}
In time $O(n \poly(\log n,\frac1\epsilon))$ time we can find a collection $\K(\ell)$ such that every pair of sets are mutually disjoint and  Equations~\ref{yeseqn} and \ref{noeqn} hold in Step~\ref{klstep}.
\end{lemma}

Note that the
algorithm runs in time $O(n L \poly(\log n,1/\epsilon))$ for $L+1$
levels of the discretized weights since we invoke
Lemma~\ref{weighted-oracle-helper} at most $L+1$ times. 
For the {\bf return} step in Step~\ref{firstreturn-non} in the algorithm,
observe:

{\small 
\begin{eqnarray*}
& & \sum_{i,k} x_{i(k)} \left( \sum_{j:(i,j) \in \h{E}_k} u^s_{ijk} - 2\varrho \zeta_{ik} \right) \\
& & \hspace{0.25in} = \sum_{i \in \viol(V)} \left( \sum_{k \leq \ki, k \in \pos(i)} \frac{\gamma \h{w}_{k}}{\Gamma(V)} \left( \sum_{j:(i,j) \in \h{E}_k} u^s_{ijk} - 2\varrho \zeta_{ik} \right) + \sum_{k > \ki, k \in \pos(i)} \frac{\gamma \h{w}_{\ki}}{\Gamma(V)} \left( \sum_{j:(i,j) \in \h{E}_k} u^s_{ijk} - 2\varrho \zeta_{ik} \right)
\right) \\
& & \hspace{0.25in} \mbox{(Using Definitions of $x_{i(\ell)},\Delta(i,\ell),\viol(V)$ in Steps~(\ref{xildef}),(\ref{deltaildef}) and (\ref{violvdef}) respectively.)}\\
& & \hspace{0.25in} =\frac{\gamma}{\Gamma(V)}  \sum_{i \in \viol(V)} \Delta(i,\ki) = \gamma \quad
\mbox{(Definition of $\Gamma(V)$ in Step~(\ref{gammavdef}))}
\end{eqnarray*}
}{\small \begin{eqnarray*}
 \sum_{i} x_{i} b_i & = & \sum_{i \in \viol(V)} \frac{\gamma \h{w}_\ki}{\Gamma(V)} b_i 
= \frac{\beta}{\Gamma(V)} \sum_{i \in \viol(V)} \frac{\gamma b_i \h{w}_{\ki}}{\beta} \leq  \sum_{i\in \viol(V)} \Delta(i,\ki) \quad
\mbox{(Definition of $\ki$)}\\
& = & \frac{\beta}{\Gamma(V)} \Gamma(V) =\beta \quad
\mbox{(Definition of $\Gamma(V)$)}
\end{eqnarray*}
}
Observe that each $x_{i(\ell)} \leq \frac{24}{\epsilon} \h{w}_\ell$.
Therefore we satisfy Lemma~\ref{weighted-oracle} in this case. In the remainder 
of the the proof we assume $\Gamma(V) \leq \epsilon \gamma/24$. Observe that:
{\small
\begin{align*}
\gamma' & = \sum_k \h{w}_k \left( \sum_{(i,j) \in \h{E}_k} 
u^s_{ijk} - 3 \varrho \sum_{i}  \bar{\zeta}_{ik} \right) = \gamma - 3 \sum_{i \in \viol(i)} \left( \sum_{k \in \pos(i), k\leq \ki}  \h{w}_k \left( \bar{\zeta}_{ik} - \zeta_{ik} \right) \right)\\
& = \gamma - \frac32\sum_{i \in \viol(i)} \sum_{k \in \pos(i), k\leq \ki}  \h{w}_k \left(
\sum_{j : (i,j) \in \h{E}_k} u^s_{ijk} - 2\varrho \zeta_{ik} \right) \geq \gamma -  \frac32\sum_{i \in \viol(i)} \Delta(i,\ki) \geq
\gamma - \frac32\Gamma(V) \geq \gamma - \epsilon \gamma/16 
\end{align*}
}
and therefore as a consequence $\gamma' \geq (1-\epsilon/16)\gamma$ 
as stated in Step~(\ref{gammapdef}) in Algorithm~\ref{alg:micro3}.
Further, for every $i$,
{\small
\begin{align*}
& \frac{\gamma b_i \h{w}_{\ki+1}}{\beta} \geq \sum_{k \in \pos(i), k > \ki} \h{w}_{\ki+1} \left(\sum_{j : (i,j) \in \h{E}_k} u^s_{ijk} - 2\varrho \zeta_{ik} \right) \quad \implies \quad
\frac{\gamma b_i }{\beta} \geq \sum_{k \in \pos(i), k > \ki} \left(\sum_{j : (i,j) \in \h{E}_k} u^s_{ijk} - 2\varrho \zeta_{ik} \right)
\end{align*}
}

\begin{algorithm}[H]
{\small
 \begin{algorithmic}[1]\parskip=0in 
\STATE 
Let $\gamma = \sum_k \h{w}_k \left( \sum_{(i,j) \in \h{E}_k} 
u^s_{ijk} - 3\varrho\sum_{i}\zeta_{ik} \right)$. Note $\varrho>0, 0<\epsilon\leq \frac1{16}$. Also $\gamma>0$ below.

\STATE Define for all $i$,
$\displaystyle \pos(i) = \left \{ k \left|
\sum_{j : (i,j) \in \h{E}_k} u^s_{ijk} - 2\varrho \zeta_{ik} > 0 \right. \right\}$ and for all $i,\ell$
{\small
\begin{align*}
\Delta(i,\ell) & = \sum_{k \in \pos(i),\ell\ge k} \h{w}_k \left(\sum_{j : (i,j) \in \h{E}_k} u^s_{ijk} - 2\varrho \zeta_{ik} \right) + \sum_{k \in \pos(i),k > \ell} \h{w}_\ell \left( \sum_{j : (i,j) \in \h{E}_k} u^s_{ijk} - 2\varrho \zeta_{ik} \right)
\end{align*}
}
\label{deltaildef}
\STATE Let $\ki=\argmax_\ell \Delta(i,\ell) > \frac{\gamma b_i \h{w}_{\ell}}{\beta}$ and $\ki=-1$ if no $\ell$ exists.
\STATE $\viol(V)=\{ i | \ki \geq 0 \}$ and $\Gamma(V) = \sum_{i \in \viol(V)} \Delta(i,\ki)$
\label{violvdef} \label{gammavdef}
\IF {$\Gamma(V) \geq \epsilon\gamma/24$} 
\STATE {\bf for} all $i \in \viol(V)$, (implies $\ki \geq 0$) 
set 
$x_{i(\ell)} = \left\{ \begin{array}{ll} \gamma\h{w}_{\ki}/\Gamma(V) & \ell \in \pos(i), \ell>\ki \\
\gamma\h{w}_{\ell}/\Gamma(V) & \ell \in \pos(i), \ell\leq\ki 
\end{array} \right. $ \label{xildef}
\STATE Set $x_{i(\ell)}=0$ for $i \not \in \viol(V)$ or $\ell \not \in \pos(i)$, all $z_{U,\ell}=0$ and {\bf return}.
\label{firstreturn-non}
\ENDIF
\STATE ${\small \bar{\zeta}_{ik} \leftarrow \left\{ \begin{array}{ll} \displaystyle \frac1{2\varrho} \hspace{-0.1in} \sum_{j : (i,j) \in \h{E}_k} u^s_{ijk} & \mbox{ if $i \in \viol(V),k \leq \ki, k \in \pos(i)$}\\ 
\zeta_{ik} & \mbox{ otherwise} \end{array} \right. }$. 
\label{barzdef}
\STATE 
Let $\gamma' = \sum_k \h{w}_k \left( \sum_{(i,j) \in \h{E}_k} 
u^s_{ijk} - 3\varrho\sum_{i}\bar{\zeta}_{ik} \right)$. 
\label{gammapdef}
\FOR {level $\ell=L$ downto $0$} 
\STATE Find a collection of sets $\K(\ell)$ using Lemma~\ref{weighted-oracle-helper} such that any two sets in $\K(\ell)$ are disjoint and 
for all $U \in \K(\ell)$: \label{klstep}
\vspace{-0.2in}
{
\small
\begin{align}
& \sum_{k \geq \ell} \left( \sum_{(i,j) \in \h{E}_k :i,j \in U}  u^s_{ijk} - \sum_{i \in U} \varrho \bar{\zeta}_{ik} \right) \geq \frac{\gamma \left\lfloor \bnorm{U}/2 \right\rfloor }{(1-\epsilon/4)\beta} \label{yeseqn}
\end{align}
}
And for all $U' \cap  \left( \cup_{U \in \K(\ell)} U \right) =\emptyset$,
{\small
\begin{align}
&\sum_{k \geq \ell} \left( \sum_{(i,j) \in \h{E}_k:i,j \in U'} u^s_{ijk} - \sum_{i \in U'} \varrho \bar{\zeta}_{ik} \right) \leq \gamma\frac{\left\lfloor \bnorm{U'}/2 \right\rfloor + \frac{\epsilon}{2}}{(1-\epsilon/4)\beta}  
\label{noeqn}
\end{align}
}
\STATE For $U \in \K(\ell)$ define 
$\Delta(U,\ell) = \sum_{k \geq \ell)} \left( \sum_{(i,j) \in \h{E}_k, i,j, \in U} u^s_{ijk} - \varrho \sum_{i \in U} \bar{\zeta}_{ik} \right)$
\ENDFOR
\STATE Define $\Gamma(\O_s) = \sum_{\ell} \sum_{U \in \K(\ell)} \Delta(U,\ell) \h{w}_{\ell}$. \label{gammaosdef}
\IF {$\Gamma(\O_s) \geq \epsilon\gamma'/24$  (Note use of $\gamma'$.) }
\STATE For $U \in \K(\ell)$ set $ z_{U,\ell} = \gamma' \h{w}_{\ell}/{\Gamma(\O_s)}$ otherwise $z_{U,\ell}=0$. (Note use of $\gamma'$.)  \label{zuldef}
\STATE Set $z_U = \max_{\ell} z_{U,\ell}$ for all $U$. All $x_{i(\ell)}$ are set to $0$ and {\bf return}.
\label{secondreturn-non}
\ENDIF
\STATE For $i \in U \in \K(\ell)$, set $\h{\zeta}_{i\ell} \leftarrow \bar{\zeta}_{i\ell} + \frac{b_i\gamma}{2\varrho \beta}\Delta(U,\ell)$, otherwise $\h{\zeta}_{i\ell} \leftarrow \bar{\zeta}_{i\ell}$.
\label{hatzetadef}
\STATE Set
$y_{ij} \leftarrow \frac{(1-\epsilon/4)\beta}{(1+\epsilon/2) \gamma} u^s_{ijk}$. Set $\mu_{ik}=\frac{(1-\epsilon/4)\beta\varrho} {(1+\epsilon/2) \gamma}\h{\zeta}_{ik}$.
 {\bf return $\{y_{ij}\},\{\mu_{ik}\}$}.
\label{finalreturn-non}
 \end{algorithmic}
 \caption{The part (ii) of {\sc MicroOracle} for matching\label{alg:micro3}}
}
\end{algorithm}
Which further implies that for any set $S$ of indices (adding nonpositive quantities to the RHS):
{\small
$$ \frac{\gamma b_i}{\beta} \geq \sum_{k \in S, k > \ki} \left(\sum_{j : (i,j) \in \h{E}_k} u^s_{ijk} - 2\varrho \zeta_{ik} \right)$$}

Now $\bar{\zeta_{ik}}=\zeta_{ik}$ for all $i \not \in \viol(V)$, $k > \ki$ or $k \not \in \pos(i)$.  Otherwise we increase $\zeta_{ik}$ to $\bar{\zeta_{ik}} = \frac1{2\varrho} \sum_{j : (i,j) \in \h{E}_k} u^s_{ijk}$ and those terms occur with a negative sign. Therefore for any $i$ and any set $S$ of indices

{\small
\begin{equation}
\label{canapply}
\frac{\gamma b_i}{\beta} \geq \sum_{k \in S} \left(\sum_{j : (i,j) \in \h{E}_k} u^s_{ijk} - 2\varrho \bar{\zeta}_{ik} \right)
\end{equation}
}

\noindent Observe that we have already proven the lemma for bipartite graphs! We now observe that the {\bf return} statement in 
Step~(\ref{secondreturn-non}) in Algorithm~\ref{alg:micro3} satisfies: 

{\small
\begin{align*}
\sum_{\ell} \sum_{U \in \O_s}  & z_{U,\ell}    \left( \sum_{k \geq \ell} \left(\sum_{(i,j) \in \h{E}_k,i,j \in U} u^s_{ijk} - \varrho \sum_{i\in U} \zeta_{ik} \right) \right) \\
& \geq \sum_{\ell} \sum_{U \in \O_s} z_{U,\ell}  \left( \sum_{k \geq \ell} \left(\sum_{(i,j) \in \h{E}_k,i,j \in U} u^s_{ijk} - \varrho \sum_{i\in U} \bar{\zeta}_{ik} \right) \right)  \quad \mbox{(since $\zeta_{ik} \leq \bar{\zeta}_{ik}$ for all $i,k$.)}\\
& = \sum_\ell \sum_{U \in \K(\ell)} \frac{\h{w}_{\ell} \gamma'}{\Gamma(\O_s)}
\left( \sum_{k \geq \ell} \left( \sum_{(i,j) \in \h{E}_k,i,j \in U} u^s_{ijk}- \varrho \sum_{i\in U} \bar{\zeta}_{ik} \right)\right)   \quad \mbox{(Definition of $z_{U,\ell}$ in Step~(\ref{zuldef}).)} \\
& = \frac{\gamma'}{\Gamma(\O_s)} \sum_{\ell} \h{w}_{\ell} \sum_{U \in \K(\ell)} \Delta(U,\ell)  = \frac{\gamma'}{\Gamma(\O_s)} \Gamma(\O_s) \geq \left( 1 - \frac{\epsilon}{16} \right) \gamma \quad \mbox{(Using Steps~(\ref{gammapdef}), (\ref{gammaosdef}).)} \\
\sum_{\ell, U \in \O_s}  &z_{U,\ell} \left\lfloor \bnorm{U}/2 \right\rfloor  =
 \sum_{\ell} \sum_{U \in \K(\ell)} \frac{\h{w}_{\ell} \gamma'}{\Gamma} \left\lfloor \bnorm{U}/2 \right\rfloor\\
& \leq \sum_{\ell} \sum_{U \in \K(\ell)} \frac{\h{w}_{\ell} \gamma'}{\Gamma(\O_s)} \frac{\beta}{\gamma} \sum_{k \geq \ell} \left( \sum_{(i,j) \in \h{E}_k :i,j \in U}  u^s_{ijk} - \varrho \bar{\zeta}_{ik} \right) quad \mbox{(Due to Equation~\ref{yeseqn}.)}\\
& \leq \frac{\beta}{\Gamma(\O_s)} \sum_{\ell} \h{w}_{\ell} \sum_{U \in \K(\ell)} \Delta(U,\ell) = \frac{\beta}{\Gamma(\O_s)} \Gamma(\O_s) =\beta \quad \mbox{($\gamma' \leq \gamma$ and Definition of $\Delta(U,\ell),\Gamma(\O_s)$)}
\end{align*}
}
Again observe that $z_{U,\ell} \leq 24\h{w}_{\ell}/\epsilon$, and that the $\K(\ell)$ are mutually disjoint. 
Therefore for any fixed $i$, $\sum_{U: i \in U} z_{U,\ell} \leq 24\h{w}_{\ell}/\epsilon$ as well. Thus for any $i$,

{\small
\[ \sum_{\ell \leq k} \left(\sum_{U: i \in U} z_{U,\ell} \right) \leq \frac{24}{\epsilon} \h{w}_{k} \left( 1 + \frac1{1+\epsilon} + \frac1{(1+\epsilon)^2} \cdots \right) \leq \frac{24}{\epsilon^2}\h{w}_k
\]
} 
Finally, based on simple accounting of the edges in a cut,
{\small
\begin{align}
&\displaystyle 2\sum_{(i,j) \in \h{E}_k: i,j \in U} u^s_{ijk} +  \sum_{i
    \in U}   \left( \sum_{j \not \in U, (i,j) \in \h{E}_k} u^s_{ijk}\right)   = \sum_{i
    \in U}   \left( \sum_{j:(i,j) \in \h{E}_k} u^s_{ijk} \right ) \qquad \forall U,k  \nonumber
\end{align}
}
which implies that 
{\small
\begin{align}
&  \sum_{k \geq \ell} \left( 2\sum_{(i,j) \in \h{E}_k: i,j \in U} u^s_{ijk} +  \sum_{i
    \in U}   \left( \sum_{j \not \in U, (i,j) \in \h{E}_k} u^s_{ijk}\right) \right) =\sum_{k \geq \ell}  \left( \sum_{i
    \in U}   \left( \sum_{j:(i,j) \in \h{E}_k} u^s_{ijk} \right) \right)  \label{accounting}
\end{align}
}
Therefore 
if $U,\ell$ were to violate $\G(\bu^s,\bx)$ then,
{\small
\begin{align*}
\sum_{k \geq \ell} \left( \sum_{(i,j) \in \h{E}_k, i,j \in U} u^s_{ijk}  -  \sum_{i \in U} \left( \sum_{j \not \in U, (i,j) \in \h{E}_k} u^s_{ijk} \right) \right) < 0 
\end{align*}
}
which along with Equation~\ref{accounting} implies
{\small
\begin{equation}
3 \sum_{k \geq \ell} \left( \sum_{(i,j) \in \h{E}_k: i,j \in U} u^s_{ijk}\right) < \sum_{k \geq \ell}\left(  \sum_{i
    \in U}\left( \sum_{j:(i,j) \in \h{E}_k} u^s_{ijk}\right)\right) \label{notlikely}
\end{equation}
}

\noindent
But from Equation~\ref{yeseqn}, if $U \in \K(\ell)$ (which must happen if $z_{U,\ell}>0$) then since $\bnorm{U} \geq 3$,
{\small
\begin{eqnarray}
& & \sum_{k \geq \ell} \left( \sum_{(i,j) \in \h{E}_k :i,j \in U}  u^s_{ijk} - \sum_{i \in U} \varrho \bar{\zeta}_{ik} \right) \geq  \frac{\gamma}{\beta} \left\lfloor \frac{\bnorm{U}}{2} \right\rfloor  \geq \frac{\gamma}{\beta} \frac{\bnorm{U}}{3} = \frac13 \sum_{i \in U} \frac{\gamma b_i}{\beta} \nonumber \\
& & \hspace{0.5in} \geq  \frac13\sum_{i \in U} \sum_{k \geq \ell} \left(\sum_{j : (i,j) \in \h{E}_k} u^s_{ijk} - 2\varrho \bar{\zeta}_{ik} \right) \quad \mbox{(Using Equation~\ref{canapply}.)} \label{need0}
\end{eqnarray}
}
But since $\bar{\zeta}_{ik} \geq 0$ we have
{\small
\begin{eqnarray*} & & 
3 \sum_{k \geq \ell} \left( \sum_{(i,j) \in \h{E}_k :i,j \in U}  u^s_{ijk}\right) - 2 \sum_{k \geq \ell} \left( \sum_{i \in U} \varrho \bar{\zeta}_{ik} \right)   \geq  3 \sum_{k \geq \ell} \left( \sum_{(i,j) \in \h{E}_k :i,j \in U}  u^s_{ijk} - \sum_{i \in U} \varrho \bar{\zeta}_{ik} \right) \\
& & \hspace{0.5in} \geq  \sum_{i \in U} \sum_{k \geq \ell} \left(\sum_{j : (i,j) \in \h{E}_k} u^s_{ijk} - 2\varrho \bar{\zeta}_{ik} \right) \quad \mbox{(Using Equation~\ref{need0}.)}
\end{eqnarray*}
} 
which contradicts Equation~\ref{notlikely}. Therefore for all $z_{U,\ell}>0$, the constraints $\G(\bu^s,\bx)$ hold. We satisfy Lemma~\ref{weighted-oracle} in this case as well.

\smallskip
 In the remainder of the the proof we assume $\Gamma(\O_s) \leq \epsilon^2 \gamma/64$. and show that the constraints of \ref{nicerlp} are satisfied.
Observe that a consequence of Equation~\ref{canapply}, even if we increase $\bar{\zeta}_{ik}$ to $\hat{\zeta}_{ik}$, we continue to satisfy:
{\small
\begin{equation}
\label{p1eqn}
\sum_{k \in S} \left(\sum_{j : (i,j) \in \h{E}_k } y_{ij} - 2\mu_{ik}\right) \leq b_i
\end{equation}
} 
We set $y_{(k)}=\max\{0, \sum_{j : (i,j) \in \h{E}_k} u^s_{ijk} - 2\varrho \hat{\zeta}_{ik} \}$.
Setting $S=\{k : y_{i(k)}>0 \}$ in Equation~\ref{p1eqn} satisfies the constraint $\sum_k y_{i(k)} \leq b_i$ for vertex $i$ in \ref{nicerlp}. 
Moreover for $i \in \left( \cup_{U \in \K(\ell)} U \right)$, $\h{\zeta}_{i\ell}$ increased from $\bar{\zeta}_{i\ell}$ by $\frac{\gamma b_i}{2\varrho\beta}$ -- thus using the observation in Step~(\ref{barzdef}) for $i \in \left( \cup_{U \in \K(\ell)} U \right)$ we have:
{\small
$$ \sum_{k \geq \ell}\left(\sum_{j : (i,j) \in \h{E}_k} u^s_{ijk} - 2\varrho \h{\zeta}_{ik} \right) \leq 0 $$
}
which implies
{\small
\begin{equation}\label{p2eqn}  \sum_{k \geq \ell}\left(\sum_{j : (i,j) \in \h{E}_k } y_{ij} - 2\mu_{ik}\right) \leq 0
\end{equation}
}
which as we will shortly see, corresponds to the vertices having no effect on feasibility of \ref{nicerlp}. 
For an $U' \in \O_s$ and $\ell$, if $U' \cap  \left( \cup_{U \in \K(\ell)} U \right) =\emptyset$ then from Equation~\ref{noeqn}:

{\small
\begin{align*}
(1+\epsilon/2) \left( \sum_{k \geq \ell} \left( \sum_{(i,j) \in \h{E}_k:i,j \in U'} y_{ij} - \sum_{i \in U} \mu_{ik} \right)\right) \leq \left\lfloor \bnorm{U'}/2 \right\rfloor + \frac{\epsilon}{2} 
\end{align*}
}
which implies that the constraint corresponding to $U',\ell$ is satisfied in \ref{nicerlp}.
If on the other hand, $U' \cap \left( \cup_{U \in \K(\ell)} U \right) = U'' \neq \emptyset$ then
{\small
\begin{align*}
& \sum_{k \geq \ell} \left( \sum_{(i,j) \in \h{E}_k:i,j \in U'} y_{ij} - \sum_{i \in U'} \mu_{ik} \right) = \sum_{k \geq \ell} \left( \frac12 \sum_{i \in U'} \left( \sum_{j \in U',(i,j) \in \h{E}_k } y_{ij} \right) 
- \sum_{i \in U'} \mu_{ik} \right) \\
&  \hspace{0.5in} = \frac12 \sum_{i \in U'} \left( \sum_{k \geq \ell} \left( \sum_{j \in U', (i,j) \in \h{E}_k} y_{ij} - 2\mu_{ik} \right) \right) \leq \frac12 \sum_{i \in U'} \left( \sum_{k \geq \ell} \left( \sum_{j:(i,j) \in \h{E}_k} y_{ij} - 2\mu_{ik} \right) \right)\\
\end{align*}
}
which using Equations~\ref{p2eqn} and \ref{p1eqn} implies that
{\small
\begin{align*}
 \sum_{k \geq \ell} \left( \sum_{(i,j) \in \h{E}_k:i,j \in U'} y_{ij} - \sum_{i \in U'} \mu_{ik} \right) 
&\leq \frac12 \sum_{i \in U' - U''} \left( \sum_{k \geq \ell} \left( \sum_{j:(i,j) \in \h{E}_k} y_{ij} - 2\mu_{ik} \right) \right) \\
&
\leq \frac12 \sum_{i \in U' - U''} b_i  
= \frac{\bnorm{U}-\bnorm{U''}}{2} \leq \frac{\bnorm{U}-1}{2} = \left \lfloor \frac{\bnorm{U}}{2} \right \rfloor
\end{align*}
}

which implies that the constraint corresponding to $U',\ell$ is also satisfied in \ref{nicerlp}. We now focus on the last remaining constraint of \ref{nicerlp}:

{\small
\begin{eqnarray*}
\sum_k \h{w}_k & &\left( \sum_{(i,j) \in \h{E}_k} 
u^s_{ijk} - 3\varrho\sum_{i} \h{\zeta}_{ik} \right) 
= \gamma' - 3\sum_{\ell} \h{w}_{\ell} \sum_{U \in \K(\ell)} \sum_{i \in U} \frac{\gamma b_i}{2\beta} = \gamma' - 3 \sum_{\ell} \h{w}_{\ell} \sum_{U \in \K(\ell)} \frac{\gamma\bnorm{U}}{2\beta} \\
& & \hspace{0.1in} \geq \gamma' - 3\sum_{\ell} \h{w}_{\ell} \sum_{U \in \K(\ell)} \frac{3\gamma\lfloor \bnorm{U}/2 \rfloor}{2\beta} \quad \mbox{(Using $\lfloor \bnorm{U}/2 \rfloor \geq \frac13 \bnorm{U}$.)} \\
& & \hspace{0.1in} \geq \gamma' - 3\sum_{\ell} \h{w}_{\ell} \sum_{U \in \K(\ell)} 
\frac{3}{2} \left( \sum_{k \geq \ell}  \left( \sum_{(i,j) \in \h{E}_k :i,j \in U}  u^s_{ijk} - \sum_{i \in U} \varrho \bar{\zeta}_{ik} \right)\right) \quad \mbox{(Using Eqn.~\ref{yeseqn}.)} \\
& & \hspace{0.1in} = \gamma' - 3 \sum_{\ell} \h{w}_{\ell} \sum_{U \in \K(\ell)}
\frac{3}{2} \Delta(U,\ell) \quad \mbox{(Using Definition of $\Delta(U,\ell)$.)}\\
& & \hspace{0.1in} \geq \gamma' - \frac92\Gamma(\O_s) \geq \gamma' - 3\epsilon \gamma/16 \geq (1-\epsilon/4)\gamma 
\end{eqnarray*}
}
Therefore
{\small
\begin{align*}
\sum_k \h{w}_k \left( \sum_{(i,j) \in \h{E}_k} 
y_{ij} - 3\sum_{i} \mu_{ik} \right) & \geq (1-\epsilon/4)\gamma \frac{(1-\epsilon/4)\beta}{(1+\epsilon/2)\gamma} \geq (1-\epsilon)\beta 
\end{align*}
}
The lemma follows. Observe that the proof was not very involved, which is the benefit of having a good relaxation \ref{nicerlp} to prove the existence of a large matching.

\section{Deferred Cut-Sparsification}
\label{sec:deferred}

\begin{definition}{\sc (The Deferred Cut-Sparsifier Problem)} 
Consider the problem: We are given a weighted graph $G=(V,E,u)$ but
the weights are {\bf not revealed} to us. Instead we are given
$\{\varsigma_{ij}\}$ with the promise that $\varsigma_{ij}/\chi \leq
u_{ij} \leq
\varsigma_{ij} \chi$. We have to produce a smaller summary data structure $\D$
storing only some of the edges. Once $\D$ is
constructed, {\bf then} the exact weights $u_{ij}$ of only those edges
stored in our structure $\D$ are revealed to us.  We then output a sparsifier $H=(V,E',u^s)$.  
\end{definition}

\begin{lemma}\label{twolevel}
Given $\{\varsigma_{ij}\}$ and
 edge weights $ \in [1/\L_0,\L_0]$,
 we create a $O(n\chi^2 \vepsilon^{-2} (\log \L_0) \log^4 n)$ size deferred sparsifier.
 We can construct a $(1\pm\vepsilon)$-sparsification from the deferred sparsification
 when the edge weights of the stored edges are revealed.

 The algorithm can be implemented in a single round of sketching; which implies 
 that it 
 can be implemented in $O(1)$ rounds in the Map-Reduce model, or a 
 single pass in the  semi-streaming model. In the semi streaming model the algorithm runs in time $O(m (\log (n/\vepsilon)) \alpha_{m,n})$ where $\alpha_{m,n}$ is the inverse Ackermann function.
\end{lemma}

\begin{proof}
The proof will follow from the fact that existing algorithms for
sparsifier construction will allow us to implement the deferred
version.  We refer the reader to the excellent survey of Fung \etal
\cite{FungHHP11} for different algorithms for construction of
sparsifiers.
The main intuition (which in some form dates back to Nagamochi and Ibaraki \cite{NI92} and used in
\cite{BenczurK96}) is that weighted sparsifiers are constructed using 
an {\em unweighted graph} by (i) first determining the probability
$p_e$ of sampling the edge $e$ and then (ii) if $e$ is sampled then
assigning it a value/importance $w_e/p_e$. For a weighted graph, part (i)
is performed for each weight class in
$\left[(1+\varepsilon)^\ell,(1+\varepsilon)^{\ell+1}\right)$.  The
space requirement is $O(n \varepsilon^{-2} (\log \L_0) \log^4 n)$ when the
weights are in the range $[1/\L_0,\L_0]$. The probability is the inverse of
the connectivity across the edge $e$ --- and this can be achieved by a
layered subsampling of edges $G_1,G_2,\ldots $ where $G_i$ contains the edges in $G_{i-1}$ 
sampled with probability $1/2$.

We make the following observation: {\em If we have the promise weight
$\varsigma_{ij}$ satisfying $u_{ij} \leq \varsigma_{ij} \chi \leq
u_{ij} \chi^2$, then if we compute $p'_e$ based on $\varsigma_{ij}$
(applying \cite{BenczurK96} sampling on the weights $\varsigma_{ij}$)
and multiply $p'_e$ by $O(\chi^2)$ then we are guaranteeing that edge
$e$ is sampled with probability at least $p_e$. }

We repeat the process for weights $\varsigma_{ij}$ in the
range $[2^{\ell},2^{\ell+1})$ for different values of $\ell$, and increase the 
sampling probability by $O(\chi^2)$. This concludes the construction of the data structure $\D$.
Observe that the sum (allowing multiple edges
between two nodes)
 of sparsifiers of a set of graphs is a sparsifier of the sum of the
 graphs. We can split the graph using $\varsigma_{ij}$
 values, construct the deferred sparsifier of each
 graph, and then take the sum of the sparsifiers (keeping one edge between the vertices).  The space requirement increases by $\log \L_0$ factor. 
The Lemma follows. 
\end{proof}

\subsection{Condition~(d\ref{defcon}) in Definition~\ref{def:amenable}, \simp\ and Deferred cut-Sparsifiers}
\label{sec:calc}

We prove a general Lemma that holds for weighted nonbipartite graphs.

\begin{lemma}\label{lemmaun3w}
Suppose we are given a graph $G=(V,E)$ where the weight of every edge $(i,j)$ is
discretized as $\h{w}_{ij} = \h{w}_{k}=(1+\epsilon)^k$ for some $k \geq 0$. 
Let $E_{k}=\{(i,j)| \h{w}_{ij}=\h{w}_{k} \}$.
Define $\set(\ell)=\{ k | k \geq \ell \}$.
Given nonnegative $\bx=\{x_{i(k)}\}\{z_{U,\ell}\}$ and nonnegative $\bu=\{u_{ijk}\}$ define:
{\small
\begin{align*}
& F(\bx,\bu) = \sum_{i,\ell} x_{i(\ell)} \left( \sum_{k:k \in \set(\ell)} \left( \sum_{j: (i,j) \in E_{k}} u_{ijk} \right) \right)  + \sum_{U \in \O_s,\ell}  z_{U,\ell} \left( \sum_{k: k \in \set(\ell)} \left( \sum_{(i,j)\in E_{k},i,j \in U} 
u_{ijk} \right) \right)
\end{align*}
}
Further suppose that:
\begin{enumerate}[(a)]\parskip=0in
\item Given any $(i,j)$, the variables $u_{ijk},u^s_{ijk}$ are only nonzero for $k$ such that $\h{w}_k=\h{w}_{ij}$.
\item For each $k\geq 0$, we have $H_k=(V,E'_k,u^s)$ as a
  $(\epsilon/16)$-Cut-Sparsifier for $G_k=(V,\h{E}_k,u)$ where $(i,j) \in E'_k$ only if $(i,j) \in \h{E}_k$. Assume $u^s_{ijk}=0$ for $(i,j) \in \h{E}_k - E'_k$. 
\item \label{defcutcond}
Let the constraint $\G(\bu^s,\bx)$ to indicate the property that for any 
$z_{U,\ell}>0$ the $\bu^s$ satisfies 
{\small
\[ \sum_{k \in \set(\ell)} \left( \sum_{(i,j) \in \h{E}_k, i,j \in U} u^s_{ijk} \right) \geq \sum_{k \in \set(\ell)} \left( \sum_{i \in U}\sum_{j \not \in U, (i,j) \in \h{E}_k} u^s_{ijk} \right) \]
}
where the RHS is $ \C(U,u^s,\set(\ell))$ and the above is equivalent to
$\G(\bu^s,\bx)$, that is either $z_{U,\ell}=0$ or 
{\small
\[\sum_{k \in \set(\ell)} \left( \sum_{(i,j) \in \h{E}_k, i,j \in U} u^s_{ijk}  - \sum_{i \in U} \left( \sum_{j \not \in U, (i,j) \in \h{E}_k} u^s_{ijk} \right) \right) \]
}
is non-negative for all $U \in \O_s,\ell$
\end{enumerate}
Then for all non-negative $\bx=\{x_{i(k)}\},\{ z_{U,\ell}\}$
  {\small
\begin{align*}
& F(\bx,\bu^s) \geq \left(1-\frac{\epsilon}{8} \right) \sum_k \sum_{(i,j)\in \h{E}_k} \h{w}_k u^s_{ijk} 
\quad \implies  \quad F(\bx,\bu) \geq \left(1-\frac{\epsilon}{2} \right) \sum_k \sum_{(i,j) \in \h{E}_k} w_k u_{ijk}
\end{align*}
}
\end{lemma}
\begin{proof}
For a cut defined by the single vertex $i$ using the sparsifier $H_k$, we have:
{\small
\begin{align}
& \left(1-\frac{\epsilon}{16}\right) \left( \sum_{j:(i,j)\in \h{E}_k} u^s_{ijk} \right)  \leq \left(\sum_{j:(i,j)\in \h{E}_k}  u_{ijk} \right)  \leq  \left(1+\frac{\epsilon}{16}\right)\left( \sum_{j:(i,j)\in \h{E}_k} u^s_{ijk}\right) \label{jeqn} 
\end{align}
}
Multiplying the above by $\h{w}_k > 0$ and summing over $k$ we have:
{\small
\begin{align}
& \left(1-\frac{\epsilon}{16}\right) \sum_k \left( \sum_{j:(i,j)\in \h{E}_k} u^s_{ijk} \h{w}_{k} \right) \leq \sum_k \left(\sum_{j:(i,j)\in \h{E}_k} u_{ijk} \h{w}_{k} \right) \leq \left(1+\frac{\epsilon}{16}\right)\sum_k \left( \sum_{j:(i,j)\in \h{E}_k} u^s_{ijk} \h{w}_{k}\right) \label{veqn}
\end{align}
}
\noindent Therefore for $\epsilon \leq 1$, using Equation~\ref{veqn},
{\small 
\begin{align}
\sum_k & \left( \sum_{(i,j)\in \h{E}_k} \h{w}_k u^s_{ijk} \right)
= \frac12 \sum_k \sum_i \left( \sum_{j:(i,j)\in \h{E}_k} u^s_{ijk} \h{w}_k  \right) \geq  \frac12 \sum_i \frac{1}{1+\frac{\epsilon}{16}}\sum_k \left( \sum_{j:(i,j)\in \h{E}_k}
u^s_{ijk} \h{w}_{k} \right) \nonumber \\
&\geq \left(1-\frac{\epsilon}{8}\right) \frac12 \sum_i 
\sum_k \left( \sum_{j:(i,j)\in \h{E}_k} u_{ijk} \h{w}_{k} \right) 
= \left(1-\frac{\epsilon}{8}\right) \sum_k  \left( \sum_{(i,j) \in \h{E}_k} \h{w}_k u_{ijk} \right) \label{wsparsifier}
\end{align}
}
Now observe that for any $U,\ell$ with $z_{U,\ell}>0$,
{\small
\begin{align}
\sum_{k \in \set(\ell)} & \left( \sum_{(i,j)\in \h{E}_k, i,j \in U} u_{ijk} \right) & = \frac12 \left[ \sum_{i \in U} \sum_k \left( \sum_{j:(i,j)\in \h{E}_k} u_{ijk} \right) - \C(U,u,\set(\ell)) \right] \label{jj1} \\
\sum_{k \in \set(\ell)} & \left( \sum_{(i,j)\in \h{E}_k, i,j \in U} u^s_{ijk} \right) & = \frac12 \left[ \sum_{i \in U} \sum_k \left( \sum_{j:(i,j)\in \h{E}_k} u^s_{ijk} \right) - \C(U,u^s,\set(\ell)) \right] \label{jj2} 
\end{align}
}
Moreover 
$\left(1-\frac{\epsilon}{16} \right) \C(U,u^s,\set(\ell)) \leq \C(U,u,\set(\ell)) \leq \left(1+\frac{\epsilon}{16} \right) \C(U,u^s,\set(\ell))$; thus from Equations~\ref{jj1} and (second part of) \ref{jeqn},

{\small 
\begin{align*}
\sum_{k \in \set(\ell)}  \left( \sum_{(i,j)\in \h{E}_k, i,j \in U} u_{ijk} \right)   
& = \frac12 \left[ \sum_{i \in U} \sum_k \left( \sum_{j:(i,j)\in \h{E}_k} u_{ijk} \right) - \C(U,u,\set(\ell)) \right] \\
& \geq  \frac12  \sum_{i \in U} \left(1-\frac{\epsilon}{16}\right) 
\sum_k \left( \sum_{j:(i,j)\in \h{E}_k} u^s_{ijk} \right)  - \frac12\left(1+\frac{\epsilon}{16} \right) \C(U,u^s,\set(\ell))  \\
& \geq \left(1-\frac{\epsilon}{16}\right) 
\left[ \sum_{i \in U} 
\sum_k \left( \sum_{j:(i,j)\in \h{E}_k} u^s_{ijk} \right) - \C(U,u^s,\set(\ell)) \right]  - \frac{\epsilon}{16} \C(U,u^s,\set(\ell)) \\
& = \left(1-\frac{\epsilon}{16}\right) \sum_{k \in \set(\ell)} \left( \sum_{(i,j)\in \h{E}_k, i,j \in U} u^s_{ijk} \right) - \frac{\epsilon}{16} \C(U,u^s,\set(\ell)) \quad \mbox{(Equation~(\ref{jj2}))} \\
& \geq \left(1- \frac{\epsilon}{8}\right) \sum_{k \in \set(\ell)} \left( \sum_{(i,j)\in \h{E}_k, i,j \in U} u^s_{ijk} \right) \qquad \mbox{(Given condition~(\ref{defcutcond}))}
\end{align*}
}
\noindent Therefore we have
{\small
\begin{align}
& \sum_{U \in \O_s,\ell} z_{U,\ell}  \left( \sum_{k \in \set(\ell)} \left( \sum_{(i,j) \in \h{E}_k,i,j \in U} \hspace{-0.1in}
u_{ijk} \right) \right)  \geq \left(1- \frac{\epsilon}{8}\right) \sum_{U \in \O_s,\ell} z_{U,\ell} \left( \sum_{k \in \set(\ell)} \left( \sum_{(i,j) \in \h{E}_k,i,j \in U}  \hspace{-0.1in}
u^s_{ijk} \right) \right)
\label{partzu}
\end{align}
}
\noindent
Using the sparsifier $H_k$ for each $k$ at cuts defined by the single vertices, for $k \in \set(\ell)$ and $x_{i(\ell)}\geq 0$, 
{\small
\begin{align}
\sum_{k \in \set(\ell)} \left( \sum_{j: (i,j) \in \h{E}_k} u_{ijk} \right) & \geq \left(1- \frac{\epsilon}{16}\right) \sum_{k \in \set(\ell)} \left( \sum_{j: (i,j) \in \h{E}_k} u^s_{ijk} \right) 
\geq \left(1- \frac{\epsilon}{8}\right) \sum_{k \in \set(\ell)} \left( \sum_{j: (i,j) \in \h{E}_k} u^s_{ijk} \right) 
\label{part1minusk}
\end{align}
}
Equation~\ref{part1minusk} multiplied by $x_{i(\ell)}\geq 0$ and summed over different $i,\ell$ gives:
{\small
\begin{align*}
& \sum_{i,\ell} x_{i(\ell)} \sum_{k \in \set(\ell)} \left( \sum_{j: (i,j) \in \h{E}_k} u_{ijk} \right) \geq \left(1- \frac{\epsilon}{8}\right) \sum_{i,\ell} x_{i(\ell)} \sum_{k \in \set(\ell)} \left( \sum_{j: (i,j) \in \h{E}_k} u^s_{ijk} \right)
\end{align*}
}
which when added to Equation~\ref{partzu}, gives
$F(\bx,\bu)  \geq (1- \frac{\epsilon}{8}) F(\bx,\bu^s)$.
But then,
{\small
\begin{align*}
F(\bx,\bu)  & \geq \left(1- \frac{\epsilon}{8}\right) F(\bx,\bu^s) \geq  \left(1- \frac{\epsilon}{8}\right)^2 
\sum_k \left( \sum_{(i,j): \in \h{E}_k}  \h{w}_k u^s_{ijk} \right) \quad \mbox{(From the statement of Lemma)}\\
& \geq \left(1- \frac{\epsilon}{8}\right)^3 
\sum_k \left( \sum_{(i,j) \in \h{E}_k}  \h{w}_k u_{ijk} \right) \qquad \mbox{(Using Equation~\ref{wsparsifier})}
\end{align*}
}
\noindent The lemma follows. 
\end{proof}

\subsection{(Deferred) Sparsification in specific models}

\noindent{\bf The Semi-Streaming Model.}
We revisit the algorithm in \cite{AhnGM12PODS}, see also
\cite{FungHHP11}. The running time follows from inspection, with the
added twist that the number of edges in Step 13 of
Algorithm~\ref{alg:sparsify} decrease geometrically. The overall
running time is $O(m (\log (n/\vepsilon))\alpha_{m,n})$ across all
edges. Note that if the edge set is partitioned across different weight classes then the running time to simultaneously compute a sparsifier for each class can also be bound by the same term.

\begin{algorithm}
{\small
 \begin{algorithmic}[1]
  \STATE Let $G_0=G$. Let $G_i$ be obtained from 
   sampling edges from $G_{i-1}$ with prob. $1/2$.
  \STATE Let $k=O(\frac{1}{\eps^2}\log^2 n)$.
  \FOR{$i=1$ to $\log_2 m$ in parallel}
   \STATE Let $F^i_1=F^i_2=\cdots=F^i_k=\emptyset$. Initialize $k$ union find structures $UF^i_1,UF^i_2,\cdots,UF^i_k$. 
   \FOR{each $e=(u,v)\in G_i$}
    \STATE Find the smallest $j$ such that 
     $UF^i_j.find(u)\neq UF^i_j.find(v)$.
    \STATE Let $F^i_j\leftarrow F^i_j\cup\{e\}$ and perform $UF_j.union(u,v)$.
   \ENDFOR
  \ENDFOR
  \STATE Let $H=(V,\emptyset)$.
   \FOR{$j=1$ to $k-1$}
    \FOR{each $e=(u,v)\in F_j$}
     \STATE If $i'$ is the smallest value of $i \in [0,\log_2 m]$ s.t. $UF^i_k.find(u)\neq UF^i_k.find(v)$, insert $e$ in $H$ with weight $w_e/2^{i'}$. (Note that we get a weighted graph.)    \ENDFOR
     \ENDFOR
  \RETURN $H$.
 \end{algorithmic}
}
 \caption{A Streaming Construction of Graph Sparsification ~\label{alg:sparsify}}
\end{algorithm}

\noindent{\bf The MapReduce Model.}
\label{mapreddeferred}
\newcommand{\vecx}{\mathbf{x}}
We show how the algorithm mentioned in Lemma~\ref{twolevel} can be implemented in the MapReduce model using $O(1)$ rounds. We use the fact that the algorithm in \cite{AhnGM12PODS} uses linear sketches of the vertex-edge adjacency matrix (after assigning arbitrary, lexicographic, assignment of directions). 
We also assume that the central server has enough memory to allow computation over near linear, i.e., $O(n^{1+1/p})$ space.
The algorithm proceeds as follows:

\begin{enumerate}\parskip=0in
 \item {\em 1st Round Mapper: } For each edge $(u,v)$, the mapper generates
  $O(\polylog n)$ bits of randomness which will be used later for sketching.
  Let $\mathbf{R}$ be these random bits. The mapper outputs $(u,(u,v,R))$
  and $(v,(u,v,R))$.
 \item {\em 1st Round Reducer: } Each reducer obtains a list of edges
  incident on a single vertex $u$ and random bits corresponding to those
  edges. Using this list and random bits, construct $\ell_0$-sampling
  sketches for a vector $\vecx_u$. Let $\mathcal{S}\vecx_u$ be the sketches.
  Output $(u,\mathcal{S}\vecx_u)$.
 \item {\em 2nd Round Mapper: } Given $(u,\mathcal{S}\vecx_u)$, output
  $(1,(u,\mathcal{S}\vecx_u))$ so that all values are collected into one
  reducer.
 \item {\em 2nd Round Reducer: } Now we have all sketches $\mathcal{S}\vecx_u$
  in one machine. The rest of the algorithm are identical to the post-processing
  of the sparsification algorithm in the context of dynamic graph streams \cite{AhnGM12PODS}.
\end{enumerate}

\section{Proof of Lemma~5}
\label{proof:init-weighted-bip}

Recall \ref{lpbm} (applied to the rescaled weights $\hat{w}$) and consider the new system \ref{bip-lpbm}

{\small
\[
\left.
\begin{minipage}{0.6\textwidth}
\begin{align*}
& \h{\beta} = \min \sum_i b_i x_i + \sum_{U \in \O} z_U \left \lfloor \frac{\bnorm{U}}{2} \right \rfloor \tag{\ref{lpbm}}\\
& x_i + x_{j}  + \sum_{U \in \O, i,j \in U} z_U \geq \h{w}_{k}  & \forall (i,j) \in \h{E}_k \\ 
& x_i,z_U \geq 0  & \forall i \in V, U \in \O 
\end{align*} 
\end{minipage} \quad \right|
\begin{minipage}{0.4\textwidth}
\begin{align*}
& \beta^b = \min \sum_i b_i x_i  \lptag \label{bip-lpbm} \\
& x_i + x_j \geq \h{w}_k  \qquad \forall (i,j) \in \h{E}_k \\
& x_i \geq 0 
\end{align*}
\end{minipage} 
\]
}

Observe that $\h{\beta} \leq \beta^b$ and it is well known that $\beta^b \leq \frac32\h{\beta}$ (set $x'_i \leftarrow x_i + \sum_{U: i\in U} z_U$, and observe that $\{x'_i\}$ is feasible for \ref{bip-lpbm}). 
The proof of Lemma~\ref{init-weighted-bip} now follows Lemma~\ref{maxbmatch} and~\ref{restofit}. Lemma~\ref{maxbmatch}
extends the following:

\begin{lemma}\cite[Lemma 3.1]{LattanziMSV11} \label{maximal}
If we sample edges in a graph uniformly at random with probability $q$ -- 
then with probability at least $1 - e^{-n}$, for every subgraph with at 
least $2n/q$ edges we have chosen at least one edge.
\end{lemma}

\begin{lemma}\label{maxbmatch}
We can construct a maximal $b$--Matching for any set of edges using $O(n^{1+1/p})$ space and $O(p)$ rounds of sketching.
\end{lemma}
\begin{proof} 
Lemma~\ref{maximal} was used in \cite[Lemma 3.2]{LattanziMSV11} to produce
a maximal matching in $O(p)$ phases when $q=O(n^{1+1/p}/|E_t|)$ for
each phase $t$ -- there $E_t$ was the number of edges remaining before
the start of phase $t$. There it was proven that for maximal matching, 
$|E_t|$, decreases by a factor of $n^{1/p}$ at each step. 
Now observe that for the uncapacitated $b$--matching whenever we
choose an edge $(i,j)$ we can increase its multiplicity to $\min \{
b_i,b_j\}$ and therefore saturate one or both the endpoints. Therefore
for every edge chosen we can ensure that one of the endpoints is not
in consideration any more. {\em As a consequence, the analysis of
\cite[Lemma 3.2]{LattanziMSV11}, instead of being limited to the set
of unmatched vertices, holds for the the set of vertices $i$ which do not
yet have degree equal to $b_i$.}  By the exact same analysis, the
number of edges decrease by a factor $n^{1/p}$ at each step and after
$O(1/p)$ steps are can guarantee that we have a maximal $b$--matching
to which we cannot add any more edges. Therefore the algorithm in
Section~3, \cite{LattanziMSV11} works, with the change that in step 3
we construct a maximal $b$--matching (and saturate one of the vertices
when we choose an edge), and step 4 considers the set of vertices $i$ which do not
yet have degree equal to $b_i$. We conclude with the observation that the primitives used in 
this algorithm can be implemented using sketches \cite{AhnGM12}.
\end{proof}

Lemma~\ref{restofit} now provides a complete proof of Lemma~\ref{init-weighted-bip}.

\begin{lemma}\label{restofit}
Let $\ae=2048 \epsilon^{-2}$ and $\epsilon_0 = 1 - \epsilon/256$.
Given uncapacitated $b$--matchings $\{M_k\}$, where each $M_k$ is maximal for the set of edges $\h{E}_k$, we can construct an initial 
nonnegative solution $\bx_0=\{x_i\},\{x_{i(k)}\}$ satisfying
$\frac{\beta^b}{\ae} \leq \beta_0=\bb^T\bx_0 = \sum_i b_ix_i \leq \frac{\beta^b}{4}$ and
{\small
\begin{align*}
& x_i(k) + x_{j(k)} \geq (1-\epsilon_0) \h{w}_{k}  & \forall (i,j) \in \h{E}_k \\ 
& x_{i(k)} \leq \h{w}_k & \forall i,k \\
& x_i - x_{i(k)} \geq 0  & \forall i,k
\end{align*}
}
$\{M_k\}$ can be computed 
using $n_{init}= O(n^{1+1/(2p)} \epsilon^{-1} \log B)$ space and $O(p)$ rounds of sketching.
\end{lemma}
\begin{proof}
We construct a maximal $b$--matching (using sketches) for each of the
levels using Lemma~\ref{maxbmatch}. 
If in the maximal $b$--matching
corresponding to level $k$ we have a vertex $i$ with degree $b_i$ then
we set $x_{i(k)} = rw_k$ where $r<1$ is a number computed
shortly. Finally Set $x_i =
\max_k x_{i(k)}$. 
Observe that we satisfy each edge
in the constraints $\bA$ to $x_{i(k)}+ x_{j(k)} \geq r w_k$. 
In the following we  bound $r,a_2$.

\begin{definition}
Recall that the edges are discretized to levels where the weight of an edge in level $k$ is $w_k = (1+\epsilon)^k$. Let $M_k$ be a maximal $b$--matching using only the edges in level $k$.
Let the edges in $M_k$ be $E(M_k)$.
\end{definition}

\noindent 
Consider an edge by edge accounting where each edge pays $r$ to the
endpoint $i$ which satisfies degree equal to $b_i$ in the maximal
matching for the level $k$. By maximality, one endpoint must satisfy
the condition and there are at most two endpoints; therefore

{\small
\begin{equation}
2r |E(M_k)| w_k \geq \sum_i b_i x_{i(k)} \geq r |E(M_k)| w_k 
\label{base0001}
\end{equation}
}
\begin{definition}
Define the highest $\lceil \log_{1+\epsilon} 2 \rceil$
levels to be group $1$, the next $\lceil \log_{1+\epsilon} 2 \rceil$ levels to be group $2$ and so on. Observe that the
edges in alternate groups go down by at least a factor of $2$ or more
precisely a number in the range $[2,2(1+\epsilon))$. Note for $\epsilon \leq \frac12$ we have $\frac{1}{2\epsilon} \leq \lceil \log_{1+\epsilon} 2 \rceil \leq \frac{2}{\epsilon}$.
\end{definition}
\begin{definition}
For each group $t$ we construct a maximal $b$--matching by first
adding the edges from the maximal $b$--matching corresponding to 
the highest level and then proceeding
downwards. Let this be $\M^G_t$; let its weight be $weight(\M^G_t)$.

Now consider an overall maximal $b$--matching where we first add $M_1$
and then add edges from $\M^G_t$ for increasing $t$, while maintaining
the maximality and feasibility (that degree of $i$ is at most
$b_i$). Let this final matching be $\M$.
\end{definition}
\begin{claim}
\label{a0001}
$\frac18 \sum_t weight(\M^G_t) \leq weight(\M)$.
\end{claim}
\begin{proof} Observe that the weights of the edges in the 
alternate groups go down by a factor of $2$. 
If we cannot add an edge from $M_t$ due to the fact that its
endpoint $i$ already satisfies its degree to be $b_i$, then we say 
that $i$ blocked that edge. Since the edges in alternate groups go down by powers of two, the total weight of edges blocked by $i$ is at most $2 \times (1 + \frac12 + \frac14 + \cdots ) \leq 4$ times the weight of the edges 
attached to $i$. Since each edge has two endpoints the lemma follows.
\end{proof}
\begin{claim}
\label{a0002}
$\beta^b \leq (1+\epsilon) 2 \sum_k w_k |E(M_k)|$
\end{claim}
\begin{proof}
Consider the edges in the optimum solution that are in level $k$. These edges either are present in the maximal $b$--Matching for this level, or there must be an endpoint $i$ for which we have $b_i$ edges adjacent already. Since there can be at most $b_i$ edges in the optimum solution that are in level $k$ and attached to $i$, the number of these edges is at most $2$ times the number of edges in the maximal $b$--matching for this level. Each edge in the optimum solution can lose a factor $(1+\epsilon)$ in the rounding down step and the lemma follows from summing over all $k$.
\end{proof}

\noindent {(Continuing proof of Lemma~\ref{init-weighted-bip})}
Using similar reasoning similar to the proof of Claims~\ref{a0001} and
\ref{a0002}, since each edge from $M_k$ 
which is not added to the corresponding $\M^G_t$ is blocked by edges
of higher weight (at one of the two vertices) for each $k$ in group $t$ we have:

{\small
\begin{equation}
\sum_{k:k \mbox{ in group } t} w_k |E(M_k)| \geq weight(\M^G_t) \geq \frac12 w_k |E(M_k)| \label{base0002}
\end{equation}
}

Using Equation~(\ref{base0001}) and the second part of Equation~(\ref{base0002}) we have:
{\small
\begin{align*}
\sum_{\stackrel{\tiny i: x_{i(k)}\neq 0}{\tiny k \mbox{ in group } t }} b_i x_i 
& \leq 2r \sum_{k: k \mbox{ in group } t } |E(M_k)| w_k  \leq  2r \frac1{\lceil \log_{1+\epsilon} 2 \rceil} 2 weight(\M^G_t) \leq \frac{8r}{\epsilon} weight(\M^G_t)
\end{align*}
}

\noindent and therefore $\sum_i b_i x_i \leq \frac{8r}{\epsilon} \sum_t weight(M_t) \leq \frac{64r}{\epsilon} weight(\M)  \leq \frac{64r}{\epsilon} \beta^b$ since $\M$ is a feasible solution. Therefore we set $r=\epsilon/256$ to ensure $\sum_i b_i x_i \leq \beta^b/4$.
Now using Equation~(\ref{base0001}) and the first part of Equation~(\ref{base0002}) we have:

{\small
\begin{align*}
\sum_{\stackrel{\tiny i: x_{i(k)}\neq 0}{\tiny k \mbox{ in group } t }}
b_i x_i \geq r \max_{k \mbox{ in group } t } |E(M_k)| w_k  
\geq r \frac1{\lceil \log_{1+\epsilon} 2 \rceil} \sum_{k \mbox{ in group } t} w_k |E(M_k)| \geq \frac{r\epsilon}{2} \sum_{k \mbox{ in group } t} w_k |E(M_k)| 
\end{align*}
}
\noindent Using Claim~\ref{a0002}, $\sum_i b_i x_i \geq \frac{\epsilon^2}{512} 
\sum_k w_k |E(M_k|  \geq \frac{\epsilon^2}{1024(1+\epsilon)} \beta_b$ and 
$a_2 =2048 \epsilon^{-2}$ proves the lemma.
\end{proof}

\section{Proof of Lemma~6}
\label{proof:oolemma}

Before proving Lemma~\ref{oolemma} We prove two structural characterizations of the matching polytope. Theorem~\ref{thm:bmatchinglaminar} was known for the case when all $b_i=1$, \cite[page 441--442]{Schrijver03}. Theorem~\ref{thm:exists}, to the best of our knowledge, was not known earlier.

\begin{theorem}
\label{thm:bmatchinglaminar} 
There exists an optimal solution of \ref{bm-dual}, the dual to the exact linear programming formulation \ref{lpbm} for the maximum $b$--Matching 
such that $L=\{U:z_U\neq 0\}$ is a laminar family, for any weights $w_{ij}$ and any graph $G=(V,E)$. Recall $b_i$ are integers, $\bnorm{U}=\sum_{i \in U} b_i$ and $\O=\{U \subseteq V| \bnorm{U} \mbox{ is odd }\}$.
{\small
\[\hspace*{-0.5cm}
\left.
\begin{minipage}{0.45\textwidth}
\begin{align*}
& \beta^*= \max \sum_{(i,j) \in E} w_{ij}y_{ij}\\
&\sum_{(i,j) \in \h{E}} y_{ij} \leq b_i  \qquad \forall i \tag{\ref{lpbm}}\\
&\displaystyle \sum_{(i,j) \in E:i,j \in U} y_{ij} \leq \left \lfloor \frac{\bnorm{U}}{2} \right \rfloor  \quad \forall U \in \O \\
& y_{ij} \geq 0 \qquad \qquad \forall (i,j) \in E
\end{align*}
\end{minipage} \quad \right | 
\begin{minipage}{0.55\textwidth}
\begin{align*}
& \beta^*=\min \sum_i b_i x_i + \sum_{U \in \O} \left\lfloor \bnorm{U}/2 \right\rfloor z_U \\
& \displaystyle x_i+x_j+\sum_{U \in \O; i,j\in U} z_U\geq w_{ij} \quad \forall (i,j)\in E \\
  & x_i, z_U \geq 0 \tag{\ref{bm-dual}} 
\end{align*}
\end{minipage} 
\]
}
\end{theorem}

\begin{proof}
 We follow the proof in \cite{Schrijver03}, page 441-442, 
 which considered standard matching with $b_i=1$ for all $i$. Let
 ${\mathfrak S}$ be the set of optimal solutions of \ref{bm-dual}.
 Let ${\mathfrak S}_2 \subseteq {\mathfrak S}$ be the subset of
 optimal solutions which {\bf minimize} \mbox{\boldmath $\sum_{z_U\neq
 0} z_U \bnorm{U}$} among the optimum solutions.  We choose a solution
 from ${\mathfrak S}_2$ that {\bf maximizes \boldmath $\sum_{z_U\neq
 0} z_U \bnorm{U}^2$}. This is a {\bf three} step choice.  We show
 that $L$ can be made a laminar family.

\medskip\noindent
 Suppose that $L$ is not a laminar family. Then, there exist $A,B\in
 L$ such that (a) $z_A,z_B\neq 0$ and (b) $A\cap B \neq \emptyset$,
 $A$ or $B$.  There are two cases: $\bnorm{A\cap B}$ is either even or
 odd.  In both cases, we change the solution (while preserving the
 objective value and the feasibility of the solution).
 \begin{enumerate} \item {\bf $\bnorm{A\cap B}$ is even:} Let
 $z=\min\{z_A,z_B\}$.  We reduce $z_A$ and $z_B$ by $z$ and increase
 $z_{A-B}$, $z_{B-A}$ by $z$.  Observe that both $A-B, B-A$ are
 nonempty and odd. We now increase every $x_i$ for $i \in A\cap B$ by
 $z$.  These changes preserve the feasibility and the objective value
 of the solution. On the other hand, they decrease $\sum_{z_U\neq 0}
 z_U\bnorm{U}$ because we replace $A$ and $B$ by $A-B$ and $B-A$
 (obviously, $\bnorm{A-B}<\bnorm{A}$ and $\bnorm{B-A}<\bnorm{B}$).
 Therefore, it contradicts the fact that the chosen solution belongs
 to ${\mathfrak S}_2$.  \item {\bf $\bnorm{A\cap B}$ is odd:} Let
 $z=\min\{z_A,z_B\}$.  We reduce $z_A$ and $z_B$ by $z$ and increase
 $z_{A\cup B}$, $z_{A\cap B}$ by $z$. Again, these changes preserve
 the feasibility and the objective value of the solution. Since
 $\bnorm{A\cup B}+\bnorm{A\cap B}=\bnorm{A}+\bnorm{B}$ and
 $\bnorm{A\cup B}>\bnorm{A},\bnorm{B}$, $\bnorm{A\cup
 B}^2+\bnorm{A\cap B}^2 \geq \bnorm{A}^2+\bnorm{B}^2$.  So
 $\sum_{z_U\neq 0} z_U \bnorm{U}^2$ increases which contradicts the
 fact that the solution maximizes $\sum_{z_U\neq 0} z_U \bnorm{U}^2$.
 \end{enumerate} Therefore, $L$ is a laminar family. 
\end{proof}

\noindent The next theorem is a surprising characterization of the $b$--Matching polytope, and is the core of the proof of Lemma~\ref{oolemma}.

\begin{theorem}\label{thm:exists}
For $0 < \epsilon \leq \frac{1}{16}$ suppose we have $G=(V,\h{E})$ such that the weight  $\h{w}_{ij}$ of every edge $(i,j) \in \h{E}$ is of the form $\h{w}_k=(1+\epsilon)^k$ for $k\geq 0$, and thus $\h{E}_k=\{(i,j)| \h{w}_{ij} =\h{w}_k\}$ and $\h{E} = \cup_k \h{E}_k$. Let $\h{w}_L$ be the largest weight.

Then $\t{\beta}\leq (1+\epsilon)\h{\beta}$ where $\t{\beta},\h{\beta}$ are defined by \ref{dual1} and \ref{bm-dual-discrete} respectively.
{\small
\[\hspace*{-0cm}
\left.
\begin{minipage}[c]{0.55\textwidth}
\begin{align*}
& \displaystyle\t{\beta}=\min \sum_i b_i x_i + \sum_{U \in \O_s} \left\lfloor \frac{\bnorm{U}}{2} \right\rfloor \sum_{\ell} z_{U,\ell} \lptag \label{dual1} \\
& \displaystyle  x_{i(k)} + x_{j(k)} + \sum_{\ell \leq k} \left(\sum_{U \in \O_s; i,j\in U} z_{U,\ell} \right) 
\geq \h{w}_k \\
& \qquad \qquad\qquad \qquad\qquad \qquad \qquad \qquad  \forall (i,j)\in E_k\\
& \displaystyle  2 x_{i(k)} + \sum_{\ell \leq k} \left(\sum_{U \in \O_s:i \in U} z_{U,\ell}\right) \leq 3\h{w}_k \quad \forall  i,k \\
& x_i - x_{i(k)} \geq 0 \qquad \forall i,k\\
&  x_{i(k)}, x_i, z_{U,\ell} \geq 0 
\end{align*}
\end{minipage}\quad \right|\quad
\begin{minipage}[c]{0.4\textwidth}
\begin{align*}
& \displaystyle\h{\beta}=\min \sum_i b_i x_i + \sum_{U \in \O} \left\lfloor \frac{\bnorm{U}}{2} \right\rfloor z_U \\
& \displaystyle x_i+x_j+\sum_{U \in \O; i,j\in U} z_U\geq \h{w}_{ij}\quad \forall (i,j)\in \h{E} \\
  & x_i, z_U \geq 0 \lptag \label{bm-dual-discrete} \\
\end{align*}
\end{minipage} 
\]
}
\end{theorem}

\begin{proof}
We begin with an optimum solution $\{x^*_i\},\{z^*_U\}$ of
\ref{bm-dual-discrete} with the weights $\h{w}_{ij}$ where the sets
$U$ such that $z_U >0$ define a laminar family -- such an optimum
solution is guaranteed by Theorem~\ref{thm:bmatchinglaminar}. 
Note that it is important that Theorem~\ref{thm:bmatchinglaminar} did not require integral edge weights.
We have $\h{\beta} = \sum_i b_i x^*_i + \sum_{U \in \O} \left\lfloor
\bnorm{U}/2 \right\rfloor z^*_U$.
Moreover, as a consequence of that optimality note that 
$\sum_{U: i \in U} z^*_U \leq \h{w}_L$ -- because otherwise we can decrease the 
$z^*_U$ corresponding to the smallest $U$ with $i \in U,z^*_U >0$ without 
affecting the feasibility. 
We produce a feasible solution for \ref{dual1} using the sequence of transformations given by Algorithm~\ref{alg:transform}.
Observe that as a consequence of Step~(\ref{firststep-transform}) of Algorithm~\ref{alg:transform} we have:

{\small
\begin{align}
& \h{x}_i + \h{x}_j + \sum_{U \in \O_s,i,j \in U} \hat{z}_U \geq \h{w}_{ij}, \qquad  \h{x}_i \leq \h{w}_L \quad \mbox{and}, \quad\sum_{U: i \in U} \hat{z}_U \leq \h{w}_L \label{hatxz}
\end{align}
}
A consequence of Step~\ref{xstep} is that $x_{i(k)} = \min \{\h{w}_k, \h{x}_i\}$ for all $i$. Therefore for bipartite graphs suppose we have an edge $(i,j) \in \h{E}_k$. 
{\small \[ x_{i(k)}  + x_{j(k)}  \geq \min\{\h{w}_k, \h{x}_i + \h{x}_j \}\geq
\h{w}_{ij}\]
}
Which implies that the scaling in Step~(\ref{finalstep-transform}) satisfies
$x_{i(k)}+x_{j(k)} \geq \h{w}_ij$. Moreover note that $x_{i(k)} \leq \h{w}_k$.
Now $\sum_i x_i b_i = \sum_i \h{x}_i b_i $ and for bipartite graphs the $z_{U}$ variables are absent, 
hence $\sum_i b_i x_i \leq \h{\beta}$. Therefore for bipartite graphs the lemma is true.
We know focus on nonbipartite graphs and 
make the following observations about the assignment to $z_{U,\ell}$:

\begin{algorithm}[t]
{\small
\begin{algorithmic}[1]\parskip=0.05in
\STATE ({\sc Remove Large Sets}) For every $U \not \in \O_s$, for every $i \in U$ we set $x^*_i
\leftarrow \min\{ x^*_i + z^*_U/2, \h{w}_L \}$ and $z^*_U \leftarrow 0$. This preserves the feasibility of
\ref{bm-dual-discrete} and increases the objective by at most $(1+\epsilon)$. 
Let this transformed set of variables be $\h{x}_i,\h{z}_U$. Recall that the largest weight class is $\h{w}_L$.
\label{firststep-transform}

\STATE ({\sc Produce $x'_{i(k)}$ values}) Set $x_i = \h{x}_i$. For each $k$ 
if $x_i > \h{w}_k$ set $x_{i(k)} = \h{w}_k$ otherwise (for $x_i \leq \h{w}_k$) set $x_{i(k)}=x_i$. 
\label{xstep}
For {\bf bipartite} graphs the algorithm skips to Step~(\ref{finalstep-transform}).

\STATE ({\sc Order Sets})  
Order the $U$ with $\h{z}_U>0$ in {\bf decreasing} order of $\bnorm{U}$ -- note that this is a laminar family. We will assign the $z_{U,\ell}$ 
corresponding to this order -- and we will assign all $z_{U,\ell}$ before proceeding to the next set $U'$ in this order.\label{orderstep}

\STATE Initialize all $z_{U,\ell}=0$.
Initialize $\sati=-1$ for all $i \in V$. 

\vspace{0.05in}
A value of $\sati \geq 0$
corresponds to the largest $k \geq 0$ such that $\sum_{\ell \leq k}
(\sum_{U:i \in U} z'_{U,\ell} )\geq \h{w}_k$. Notice that when we consider $U$ according to Step~(\ref{orderstep})
$\sum_{\ell \leq k} (\sum_{U':i \in U'} z_{U',\ell})$ is same for all 
$i \in U$ and the $\sati$ values are the same for all $i \in U'$. Define that value as $\satu$.

\FOR {the next set $U$ according to the order in Step~(\ref{orderstep})} \label{iterstart}
\STATE $k=\satu+1$.
\WHILE {$(k \leq L)$ and $(\h{z}_U>0)$} 
\IF {$\h{z}_U + \sum_{\ell \leq k} (\sum_{U:i \in U} z_{U,\ell} )\geq \h{w}_k$}
\STATE Set $z'_{U,k} \leftarrow \h{w}_k - \sum_{\ell \leq k} (\sum_{U:i \in U} z_{U,\ell} )$.
\STATE $\h{z}_U \leftarrow \h{z}_U - z_{U,k}$.
\STATE $\sati\leftarrow k$ for all $i \in U$ and $k \leftarrow k+1$.
\ELSE
\STATE  Set $z_{U,k} = \h{z}_U$, $\h{z}_U \leftarrow 0$.
\ENDIF
\ENDWHILE 
\ENDFOR\label{iterend}
\STATE Return $\{x_{i(\ell)}\},\{x_i\},\{z_{U,\ell}\}$. \label{finalstep-transform}
\end{algorithmic}
}
\caption{Producing a feasible solution for \ref{dual1}\label{alg:transform}}
\end{algorithm}

{\small
\[ \sum_{\ell \leq k} (\sum_{U:i \in U} z'_{U,\ell} ) =\min\{ \h{w}_k, \sum_{U:i\in U} \h{z}_U \} \]
}
Therefore for an edge $(i,j) \in E_k$:
{\small
\begin{align*}
& x_{i(k)} + x_{j(k)} + \sum_{\ell \leq k} \left( \sum_{U \in \O_s,i,j \in U} z'_{U,\ell} \right)   \geq 
\min \left \{ \h{w}_k, \h{x}_i + \h{x}_j + \sum_{U \in \O_s, i,j \in U} \h{z}_U \right \} \geq
\h{w}_{k}
\end{align*}
}
Now from Step~(\ref{firststep-transform}) :
{\small
\begin{align*}
\sum_i x_i b_i + \sum_{U \in O_s} \left \lfloor \frac{\bnorm{U}}{2} \right \rfloor \sum_{\ell} z_{U,\ell} 
& \leq 
\sum_i \h{x}_i b_i + \sum_{U \in O_s} \left \lfloor \frac{\bnorm{U}}{2} \right \rfloor \h{z}_U  \leq (1+\epsilon)\h{\beta}
\end{align*}
}
The theorem follows.
\end{proof}

\begin{nlemma}{\ref{oolemma}}
For any $0<\epsilon \leq \frac1{16}$ ,suppose we are given a subgraph $G=(V,E')$
where $|V|=n$ and the weight  $\h{w}_{ij}$ of every edge $(i,j) \in E'$ is of the form $\h{w}_k=(1+\epsilon)^k$ for $k\geq 0$, and thus $E'_k=\{(i,j)| \h{w}_{ij} =\h{w}_k\}$ and $E' = \cup_k E'_k$.
If we are also given a 
feasible solution to the system \ref{nicerlp} 
then we find an integral solution of \ref{nicerlp0}of weight $(1-2\epsilon)\beta$ 
using $O(|E'| \poly (\epsilon^{-1}, \log n))$ time using edges of $E'$.

{\small
\begin{align*}
& \sum_{k} \h{w}_k \left( \sum_{(i,j) \in E'_k} y_{ij} - 3 \sum_{i}  \mu_{ik} \right) \geq (1-\epsilon)\beta \\
&\sum_{(i,j) \in E'_k} y_{ij} - 2 \mu_{ik}\leq y_{i(k)}\forall i,k\\
&\sum_k y_{i(k)} \leq b_i  \qquad \forall i \tag{\ref{nicerlp}}\\
&\displaystyle  \sum_{k \geq \ell} \left( \sum_{(i,j) \in E'_k:i,j \in U} y_{ij} - \sum_{i \in U} \mu_{ik} \right)\leq \left \lfloor \frac{\bnorm{U}}{2} \right \rfloor  \quad \forall U \in \O_s, \forall \ell  \\
& y_{ij},y_{i(k)},\mu_{ik} \geq 0 \qquad \qquad \forall (i,j) \in E', i,k
\end{align*}
}
\end{nlemma}
\begin{proof}
Consider the modification of \ref{nicerlp} and the maximum $b$--matching restricted to $E'$:
{\small
\[\hspace*{-0.5cm}
\begin{minipage}{0.65\textwidth}
\begin{align*}
& \t{\beta}(E')=\max \sum_{k} \h{w}_k \left( \sum_{(i,j) \in E'_k} y_{ij} - 3 \sum_{i}  \mu_{ik} \right) \\
& \sum_{(i,j) \in E'_k} y_{ij} - 2 \mu_{ik} \leq y_{i(k)}  \qquad \forall i,k \\
&\sum_{k} y_{i(k)} \leq b_i  \qquad \forall i \lptag \label{nicerlp001}\\
&\displaystyle  \sum_{k \geq \ell} \left( \sum_{(i,j) \in E'_k:i,j \in U} y_{ij} - \sum_{i \in U} \mu_{ik} \right)\leq \left \lfloor \frac{\bnorm{U}}{2} \right \rfloor  \quad \forall U \in \O_s, \forall \ell  \\
& y_{ij},\mu_{ik},y_{i(k)} \geq 0 \qquad \qquad \forall (i,j) \in E', i,k
\end{align*}
\end{minipage} \left| \quad
\begin{minipage}{0.35\textwidth}
\begin{align*}
& \h{\beta}(E')= \max \sum_{(i,j) \in E'} \h{w}_{ij}y_{ij}\\
&\sum_{(i,j) \in E'} y_{ij} \leq b_i  \qquad \forall i \tag{\ref{nicerlp0}($E'$)}\\
&\displaystyle \sum_{(i,j) \in E':i,j \in U} y_{ij} \leq \left \lfloor \frac{\bnorm{U}}{2} \right \rfloor  \quad \forall U \in \O \\
& y_{ij} \geq 0 \qquad \qquad \forall (i,j) \in E'
\end{align*}
\end{minipage} \right.
\]
}
Observe that the dual of \ref{nicerlp001} is \ref{dual1} whereas the dual of \ref{nicerlp0} 
is \ref{bm-dual-discrete}, restricted to the edge sets $E'$. 
Therefore by {\bf weak duality} the existence of a solution of \ref{nicerlp} 
will imply $\t{\beta}(E')  \geq (1-\epsilon)\beta$. Therefore applying Theorem~\ref{thm:exists}
$\h{\beta}(E') \geq (1-\epsilon)\beta/(1+\epsilon) \geq  (1-2\epsilon)\beta$. 
Recall that we can use near linear time algorithms for
weighted matching \cite{DuanP10} or for weighted $b$--matching \cite{AhnG14} to find a
solution of size $(1-\epsilon)\h{\beta}(E')$ in time $O(|E(Y)|
\poly(\epsilon^{-1}\log n))$. That solution has value at least
$(1-2\epsilon)\beta$ and proves Lemma~\ref{oolemma}.
\end{proof}

\section{Proof of Lemma~9}
\label{proof:weighted-oracle-helper}

\begin{nlemma}{\ref{weighted-oracle-helper}}
In time $O(n \poly(\log n,\frac1\epsilon))$ time we can find a collection $\K(\ell)$ such that every pair of sets are mutually disjoint and  
for $U \in \K(\ell) \in \O_s$
{
\small
\begin{align}
& \frac{(1-\epsilon/4)\beta}{\gamma} \sum_{k \geq \ell} \left( \sum_{(i,j) \in \h{E}_k :i,j \in U}  u^s_{ijk} - \sum_{i \in U} \varrho \bar{\zeta}_{ik} \right) \geq \left\lfloor \frac{\bnorm{U}}{2} \right\rfloor \label{yeseqn0}
\end{align}
}
And for $U' \cap  \left( \cup_{U \in \K(\ell)} U \right) =\emptyset$,
{\small \begin{align}
&\frac{(1-\epsilon/4)\beta}{\gamma} \sum_{k \geq \ell} \left( \sum_{(i,j) \in \h{E}_k:i,j \in U'} u^s_{ijk} - \sum_{i \in U'} \varrho \bar{\zeta}_{ik} \right) \leq \left\lfloor \frac{\bnorm{U'}}{2} \right\rfloor + \frac{\epsilon}{2} \label{noeqn0}
\end{align}
}
\end{nlemma}
\begin{proof}
We first consider a graph problem where we collect a maximal
collection of dense small odd-sets.  This abstraction will also be
used later for the case when $w_{ij} \neq 1$. 
Note we will be upper bounding $n\poly(\kappa,\log n)$ by
$n^{1+1/p}$ asymptotically. We first prove:

\begin{lemma}
\label{maximaldense2}
Given $G=(V,E)$, non-negative numbers $q_{ij}$ on edges and 
$\hat{q}_i$ on nodes, such that 
\begin{enumerate}[({A}1)]\parskip=0in
\item If $b_i = 1$ then $\hq_i \geq C$.
\item $\sum_{j} q_{ij} \leq \hq_i$ for all $i$.
\item Any odd-set $U$ with $\bnorm{U} > 4/\epsilon$ satisfies $ \displaystyle \sum_{i \in U} \left( \hq_i - \sum_j q_{ij} \right) > 1$
\end{enumerate}
using space $O(n\epsilon^{-2})$ and $O(m)+O(n\poly(\epsilon^{-1},\log n))$ time and a single pass over the list of edges we can find a collection $\L$ of mutually disjoint sets such that
\begin{enumerate}[(i)]\parskip=0in
\item Every odd-set $U \in \L$ satisfies 
$ \displaystyle \sum_{(i,j):i,j\in U} q_{ij}\geq \frac{1}{2}\left(\sum_{i\in
    U}\hq_i- 1\right)
$
\item Every odd set $U \not \in \L$ either intersects with a set in $\L$ or satisfies
$ \displaystyle
\sum_{(i,j):i,j\in U} q_{ij}\leq \frac{1}{2}\left(\sum_{i\in
    U}\hq_i- (1-\epsilon)\right)
$
\end{enumerate}
\end{lemma}
\begin{proof}
We create a graph $H$ on the vertex set $V\cup\{s\}$ where we have $\lfloor q_{ij} 8 \epsilon^{-3} \rfloor$ parallel edges between $i$ and $j$. After the edges have been added we add edges between $i$ and $s$ till the degree of $s$ is 
$\lceil \hat{q}_i 8 \epsilon^{-3} \rceil$ --- this is feasible due to (A2). We use 
a lemma from \cite{AhnG14} which combines the minimum odd-cut approach in 
\cite{PadbergR82} with the construction of approximate Gomory-Hu Trees \cite{lowstcuts,BhalgatHKP07}. 

\begin{lemma}\cite[Lemma~12]{AhnG14}
\label{quotelemma}
 Given an unweighted graph $G$ with parameter $\kappa$ 
  and a special node $s$, in time $O(n\poly(\kappa,\log n))$
we can identify a collection $\L$ of mutually disjoint 
  odd-sets which (i) do not contain $s$  (ii) define cut of at most $\kappa$ in $G$
  and (iii) every other odd set not containing $s$ and with a cut less
  than $\kappa$ intersects with a set in $\L$.
\end{lemma}

We now apply Lemma~\ref{quotelemma} with $\kappa = \lfloor 8 \epsilon^{-3} \rfloor$. Note that any set which is returned in $\L$ satisfies condition (i) easily since
{\small
\begin{align*}
& \sum_{(i,j):i,j \in U} \lfloor q_{ij} 8 \epsilon^{-3} \rfloor \geq \frac12 \left( \sum_{i \in U} \lceil \hat{q}_i 8 \epsilon^{-3} \rceil -  \lfloor 8 \epsilon^{-3} \rfloor \right) \implies 
\sum_{(i,j):i,j \in U} q_{ij} \geq \frac12 \left( \sum_{i \in U} \hat{q}_i  -   1 \right)
\end{align*}
}
For any odd set which is not returned and does not intersect the any of the sets returned, the cut after discretization is at least $\kappa$;
{\small
\begin{align*}
& \sum_{(i,j):i,j \in U} \lfloor q_{ij} 8 \epsilon^{-3} \rfloor \leq \frac12 \left( \sum_{i \in U} \lceil \hat{q}_i 8 \epsilon^{-3} \rceil -  \lfloor 8 \epsilon^{-3} \rfloor \right) \implies \sum_{(i,j):i,j \in U} q_{ij}  \leq \frac{{{4/\epsilon} \choose 2}}{8\epsilon^{-3}} + 
\frac12 \left( \sum_{i \in U} \hat{q}_i  + \frac{4/\epsilon}{8\epsilon^{-3}}  -  1 \right)
\end{align*}
}
which gives us Condition (ii).
\end{proof}

\noindent Observe that (A1) and (A3) imply that singleton vertices or very large sets cannot be present in $\L$. We are now ready to prove Lemma~\ref{weighted-oracle-helper}.
The algorithm is simple:

\begin{enumerate}\parskip=0in
\item Set $\h{q}_i(\ell) = b_i + \frac{2(1-\epsilon/4)\varrho\beta}{\gamma}\sum_{k \geq \ell} \bar{\zeta}_{ik}$.
\item For $i,j$ consider the $k$ such that $u^s_{ijk}\neq 0$. Set $q_{ij}(\ell)= 
\frac{(1-\epsilon/4)\beta}{\gamma} u^s_{ijk}$ if $k \geq \ell$ and $0$ otherwise. \label{maddsop-non}
\item Let $C=1$ and apply Lemma~\ref{maximaldense2} and get a collection of mutually disjoint sets, which we denote by $\K(\ell)$.
\end{enumerate}

Equation~\ref{canapply}, proved just after the description of Algorithm~\ref{alg:micro3}, applied to $S=\{k| k\geq \ell\} $ (multiplied by $(1-\epsilon/4)\beta/\gamma)$ gives us
{\small
\begin{align*}
& (1-\epsilon/4) b_i \geq \sum_{j} q_{ij}(\ell) - \sum_{k \geq \ell }\frac{2(1-\epsilon/4)\varrho\beta}{\gamma}\bar{\zeta}_{ik} \implies \h{q}_i(\ell) \geq \sum_j q_{ij}(\ell)
\end{align*}
}

And for $b_i=1$,  $\h{q}_i(\ell) \geq b_i = C$.
Moreover for any set $U$, 
{\small
\begin{align*}
& \sum_{i \in U} \left( \h{q}_i(\ell) - \sum_j \h{q}_i(\ell) \right) = \sum_{i \in U} \left( b_i -  \frac{(1-\epsilon/4)\beta}{\gamma}\sum_{k \geq \ell)} \left(\sum_{(i,j) \in \h{E}_k} u^s_{ijk} - 2\varrho \zeta_{ik} \right)\right) \geq \sum_{i \in U} \epsilon b_i/4 =\epsilon \bnorm{U}
\end{align*}
}
Therefore if $\bnorm{U} > 4/\epsilon$ then $\sum_{i \in U} \left( \h{q}_i(\ell) - \sum_j \h{q}_i(\ell) \right)>1$.
This implies that the conditions (A1)--(A3) of Lemma~\ref{maximaldense2} are valid and we can apply Lemma~\ref{maximaldense2}. 
The lemma follows.
\end{proof}

\bibliographystyle{abbrv}

\end{document}